\documentclass[a4paper, 12pt]{scrartcl}
\usepackage[english]{babel}
\usepackage[ansinew]{inputenc}
\usepackage{amsmath}
\usepackage[T1]{fontenc}
\usepackage{enumerate}
\usepackage[all]{xy}
\usepackage{amsmath,amssymb,amsthm,enumerate}
\usepackage{graphicx}
\usepackage{epsfig,psfrag}
\usepackage{bbm}

\newtheorem{remark}{Remark}
\newtheorem{theorem}{Theorem}
\newtheorem{definition}{Definition}
\newtheorem{algorithm}{Algorithm}
\newtheorem{proposition}{Proposition}
\title{\Large The Moran model with selection: Fixation probabilities,
ancestral lines, and an alternative particle representation}
\author{\large Sandra Kluth,
Ellen Baake\footnote{Corresponding author. \newline 
E-mail address: ebaake\emph{@}techfak.uni-bielefeld.de \newline
Tel.: +49 521 106 4896 (E. Baake)} \\
\small{Faculty of Technology,
Bielefeld University,
Box 100131,
33501 Bielefeld, Germany } \\
}
\date{}

\begin{document}
 
\maketitle

\begin{abstract}
\textbf{Abstract.}
We reconsider 
the Moran model 
in continuous time with population size $N$, two allelic types, and selection.
We introduce a new particle representation, 
which we call the labelled Moran model, and which has the same distribution of type frequencies
as the original Moran model, provided the initial values are chosen appropriately.
In the new model, individuals are labelled $1,2, \dots, N$; 
neutral resampling events may take place between arbitrary labels, whereas
selective events 
only occur in the direction of
increasing labels. 
With the help of elementary methods only,
we not only recover fixation probabilities, but also obtain detailed insight into
the number and nature of the selective events that play a role in the fixation process
forward in time.

\end{abstract}

\section{Introduction}

In population genetics models for finite populations including
selection, many properties are still unknown, in particular when it
comes to ancestral processes and genealogies. This is why a
classic model, namely the Moran model
in continuous time with population size $N$, two allelic types, and (fertility) selection is
still the subject of state-of-the-art research, see e.g. Bah et al. (2012), Mano (2009),
Houchmandzadeh and Vallade (2010), or
Pokalyuk and Pfaffelhuber (2013). In particular, new insights are currently being obtained 
by drawing a more fine-grained picture: Rather than looking only
at the type frequencies as a function of time, one considers
ancestral lines and genealogies of samples 
as in the ancestral selection graph (ASG) (Krone and Neuhauser (1997), Neuhauser and Krone (1997)), or 
even particle representations
that make all individual lines and their interactions explicit (here the main tool is
the lookdown construction of
Donnelly and Kurtz (1996, 1999)). See Etheridge (2011, Ch. 5) for an excellent overview
of the area.

The way that fixation probabilities (that is, the probabilities
that the individuals of a given  type eventually take over in the 
population) are dealt with is typical of this development.
The classical result is that of Kimura (1962), which is based on
type frequencies in the diffusion limit. Short and elegant standard arguments today 
start from the discrete setting, use a first-step or martingale approach,
and arrive at the fixation probability in the form of
a (normalised) series expansion in the reproduction rate of the 
favourable type (cf. Durrett (2008, Thm. 6.1)). An alternative formulation is
given by Kluth et al. (2013). Both are easily seen to converge to Kimura's
result in the diffusion limit.  

Recently, Mano (2009) obtained the fixation probabilities with the help of the
ASG, based on methods of duality (again in the diffusion limit); his argument is nicely
recapitulated by Pokalyuk and Pfaffelhuber (2013). What is still missing is a derivation within the framework
of a full particle representation.

This is done in the present article. To this end, we introduce an alternative
particle system, which we call the \emph{labelled Moran model}, and
which is particularly well-suited for finite populations under selection.
It is reminiscent of the 
\textit{$N$-particle look-down process}, but the main new idea is that each
of the $N$ individuals is characterised by a different reproductive
behaviour, which we indicate by a label. With the help of elementary methods only,
we not only recover fixation probabilities, but also obtain detailed insight into
the number and nature of the selective events that play a role in the fixation process
forward in time.

The paper is organised as follows. We start with a brief survey of the Moran model
with selection (Sec. \ref{The Moran model with selection}).
Sec. \ref{Reflection principle} is a warm-up exercise that characterises fixation probabilities
by way of a reflection principle. 
Sec. \ref{Labelled Moran model} introduces the labelled Moran model. With the help
of a coupling argument that involves permutations of the reproductive events, we obtain the probability that a 
particular label becomes fixed. Within 
this setting, we identify reproduction events that affect the fixation probability of a given label; 
they are termed \textit{defining events}. 
The number of defining events that are \textit{selective} turns out as the pivotal quantity
characterising the fixation probabilities of the individual labels, and hence the 
label of the line that will become ancestral to the entire population.
In Sec. \ref{Diffusion limit of the labelled Moran model} we pass
to the diffusion limit. We continue with a simulation algorithm that generates
the label that becomes fixed together with the targets of the selective defining events
(Sec. \ref{A simulation algorithm}).
Sec. \ref{Discussion} summarises and discusses the results.

\section{The Moran model with selection}
\label{The Moran model with selection}
We consider a haploid population of fixed size $N \in \mathbb{N}$ in 
which each individual is characterised by a type $i \in S = \{A,B\}$. 
If an individual reproduces, its single offspring inherits the parent's 
type and replaces a randomly chosen individual, maybe its own parent. This
way the replaced individual dies and the population size remains constant. 

Individuals of type $B$ reproduce at rate $1$, whereas individuals of type 
$A$ reproduce at rate $1+s^{}_N$, $s^{}_N \geq 0$. 
Accordingly, type-$A$ individuals are termed  `more fit', type-$B$ individuals are
`less fit'. 
In line with a central idea of the ASG, we will
decompose reproduction events into \textit{neutral} and \textit{selective} ones. Neutral ones
occur at rate $1$ and happen to all individuals, whereas  selective events 
occur at rate $s_N^{}$ and are
reserved for type-$A$ individuals. (At this stage, these rates are understood as rates per individual.)

The Moran model has a well-known graphical representation as an
interacting particle system (cf. Fig. \ref{moran model}).
The $N$ vertical lines represent the $N$ individuals and time runs from 
top to bottom in the figure. Reproduction events are represented by arrows with the 
reproducing individual at the tail and the offspring at the head.

\begin{figure}[ht]
\begin{center}
\psfrag{t}{\Large$t$}
\psfrag{0}{\Large$A$}
\psfrag{1}{\Large$B$}
\includegraphics[angle=0, width=7cm, height=5cm]{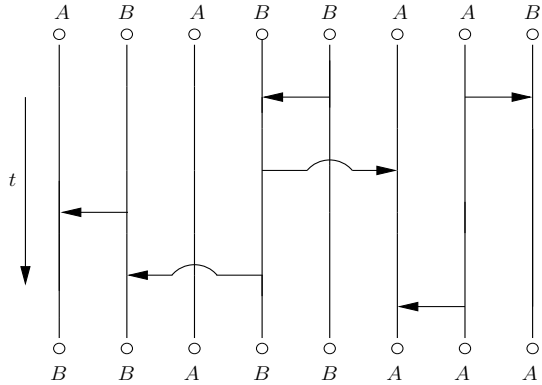}
\caption{The Moran model. The types ($A=$ more fit, $B=$ less fit) are
indicated for the initial population (top) and the final one (bottom). }
\label{moran model}
\end{center}
\end{figure} 

We now consider the process $\left( Z_{t}^N \right)_{t \geq 0}$ (or simply $Z^N_{}$), where 
$Z_{t}^N$ is the number of individuals of type $A$ at time $t$.
(Note that this is the type frequency representation, which contains
less information than the particle (or graphical) representation in Fig.~\ref{moran model}.) 
$Z^N_{}$ is a birth-death process with birth rates $\lambda_i^N$ and
death rates $\mu_i^N$ when in state $i$, where
\begin{equation}
 \label{eq:lambda_mu}
\lambda_i^{N}=(1+s_N^{}) i \frac{N-i}{N}\quad \text{and}
\quad \mu_i^{N}=(N-i)\frac{i}{N}. 
\end{equation}
The absorbing states are $0$ and $N$, thus, one of the two types, $A$ or $B$, will
almost surely become fixed in the population in finite time. 
Let $T_k^N := \min \{ t \geq 0 : Z_t^N = k  \}$, $0 \leq k \leq N$, be the first hitting time of state $k$ by $Z^N_{}$.
The \textit{fixation probability} of type $A$ given that there are initially $k$ type-$A$ individuals
is well known to be (cf. Durrett (2008, Thm. 6.1))
\begin{equation}
 h_k^{N} := \mathbb{P}( T_N^{N}< T_0^{N} \mid Z^N_0 = k ) 
 = \frac{\sum_{j=N-k}^{N-1} (1+s_N^{})^j_{}}{\sum_{j=0}^{N-1}(1+s_N^{})^j_{}}.
 \label{h_i}
\end{equation}

An alternative formulation is obtained by Kluth et al. (2013), where fixation probabilities
under selection \emph{and} mutation are considered. It is easily verified that, taking together their
Eq.~(40), Theorem 2, and Eq.~(30) and setting the mutation rate to zero
gives the representation
\begin{equation}\label{h_i_alt}
h_k^N = \frac{k}{N} +  \sum_{n=1}^{N-1} a_n^N \frac{k}{N} \prod_{j=0}^{n-1} \frac{N-k-j}{N-1-j},
\end{equation}
where the coefficients $a_n^N$ follow the recursion
\begin{equation}\label{a_n}
a_n^N - a_{n+1}^N = s_N^{} \frac{N-n}{n+1} (a_{n-1}^N - a_n^N)
\end{equation}
for $1 \leq n \leq N-2$,
with initial conditions
\begin{equation}\label{a_0}
a_0^N=1 \text{ \ and \ } a_1^N = a_0^N - N(1-h_{N-1}^N) .
\end{equation}
Eq. \eqref{h_i_alt} is suggestive: It decomposes the fixation probability into the neutral part
($k/N$) plus additional terms due to offspring created by selective events. It is the purpose 
of this paper to make this interpretation more precise, and to arrive at a thorough understanding
of Eq. \eqref{h_i_alt}. In particular, we aim at a probabilistic understanding of the coefficients
$a_n^N$.

Before doing this it is useful to consider the model in the diffusion limit. To this end, we use the usual rescaling 
\begin{equation*}
\left( X_{t}^N  \right)_{t \geq 0} := \frac{1}{N} \left( Z_{Nt}^N  \right)_{t \geq 0},
\end{equation*}
and assume that $\lim_{N \to \infty}Ns_N^{} = \sigma$, $0 \leq \sigma < \infty$. As $N \to \infty$, $ ( X_{t}^N )_{t \geq 0}$
converges in distribution to
the well known diffusion process
$\left(X_t^{} \right)_{t \geq 0}$,
which is characterised by the drift coefficient
$ a(x) = \sigma x (1-x) $
and the diffusion coefficient
$ b(x)= 2 x (1-x).$
Hence, the infinitesimal generator $\mathcal{A}$ of the diffusion is defined by
\begin{equation*}
\mathcal{A}f(x)= x (1-x)  \frac{\partial^2{}}{\partial{x^2}} f(x)  + \sigma x (1-x)  \frac{\partial{}}{\partial{x}} f(x) , 
\ f \in \mathcal{C}^2_{}([0,1]).
\end{equation*}
Define the first passage time
$T_x^{}:= \inf \{ t \geq 0 : X_t^{} = x \}$, $x \in [0,1]$, as the first time that $X_t^{}$ equals $x$. For $\sigma \neq 0$ the fixation
probability of type $A$ follows a classical result of Kimura (1962)
(see also Ewens (2004, Ch. 5.3)):
\begin{equation}
 h(x):= \mathbb{P} (T_1 < T_0 \mid X_0^{} = x)
 = \frac{1-\exp(- \sigma x)}{1 - \exp(- \sigma)}.
 \label{h}
\end{equation}

See also Ewens (2004, Ch. 4, 5) or Karlin and Taylor (1981, Ch. 15) for reviews of the diffusion 
process describing
the Moran model with selection and its
probability of fixation.

According to Kluth et al. (2013) and previous results of Fearnhead (2002)
and Taylor (2007) (obtained in the context of the common ancestor process in the 
Moran model with selection and mutation), the equivalent of \eqref{h_i_alt} reads 
\begin{equation}
 h(x)= x + \sum_{n \geq 1} a_n^{} x(1-x)^n
 \label{h(x)}
\end{equation}
with coefficients $a_n^{}= \lim_{N \to \infty} a_n^N$, $n \geq 1$. Following \eqref{a_n} and \eqref{a_0},
the $a_n^{}$ are characterised by the recursion
\begin{equation}
 a_n^{} - a_{n+1}^{} = \frac{\sigma}{n+1} (a_{n-1}^{} - a_n^{})
\label{a_n Fearnhead}
 \end{equation}
with initial conditions
\begin{equation}
 a_0^{}=1 \text{ \ and \ } a_1^{}=a_0^{}- h'(1)  .
 \label{a_0 Fearnhead}
\end{equation}

\section{Reflection principle}
\label{Reflection principle}
Durrett (2008, Ch. 6.1.1) proves equation \eqref{h_i} in two ways, namely, using
a first-step approach and a martingale argument, respectively.
Both approaches rely on the process $\left( Z^N_t \right)_{t \geq 0}$, without reference
to an underlying particle representation. As a warm-up exercise,
we complement this by an approach based on the particle picture, which we
call the \textit{reflection principle}. 
\begin{definition}
 Let a graphical realisation of the Moran model be given, with $Z_0^N = k$, $1 \leq k \leq N-1$,
 and fixation of type $A$. Now interchange the types (i.e., replace all $A$ individuals
 by $B$ individuals and vice versa), without otherwise changing the graphical realisation. This 
 results in a graphical realisation of the Moran model with 
 $Z_0^N = N-k$ in which type $B$ becomes fixed, and is called
 the reflected realisation.
\end{definition}
Put differently, in the case $Z_0^N = k$, $1 \leq k \leq N-1$,
reflection transforms a realisation in which the (offspring of the) $k$ type-$A$ individuals
become fixed (altogether, this happens with probability $h^N_k$) into a realisation in
which the (offspring of the) type-$B$ individuals become fixed (which happens with probability $1 - h^N_{N-k}$), and vice versa.
This operation does not change the graphical structure, but the weights of the realisations
are different since a different weight is attached to some of the arrows. This is best explained
in the example given in Fig. \ref{mirroring}: 
Bold lines represent type-$A$ individuals,
thin ones type-$B$ individuals; likewise, arrows emanating from type-$A$
(type-$B$) individuals are bold (thin). Interchange of types transforms the realisation
on the left into the realisation on the right and vice versa, such that
the respective other type becomes fixed. Arrows that are marked by $1$ respectively $1+s_N^{}$ appear at rate $1/N$ respectively
$(1+s_N^{}) /N$ (per \textit{pair} of individuals).
\begin{figure}[htbp]
\begin{center}
\begin{minipage}{0.4 \textwidth}
\begin{center}
\psfrag{A}{\Large$A$}
\psfrag{B}{\Large$B$}
\psfrag{t}{\Large$t$}
\psfrag{1}{\Large$1$}
\psfrag{1+s}{\Large$1+s_N^{}$}
\includegraphics[angle=0, width=6cm, height=4.5cm]{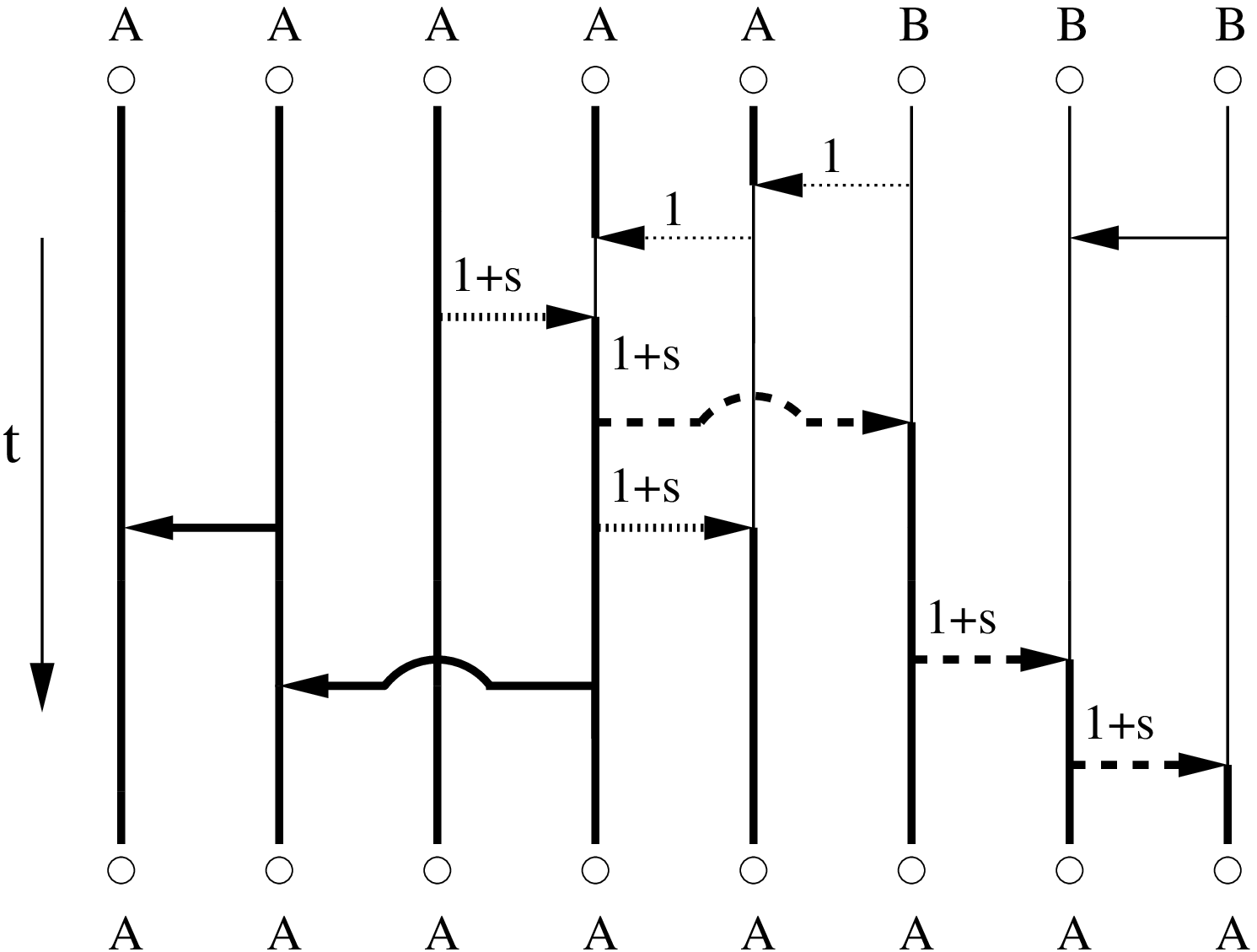}
\end{center}
\end{minipage}
\hspace{0.5cm}
\begin{minipage}{0.4 \textwidth}
\begin{center}
\psfrag{A}{\Large$A$}
\psfrag{B}{\Large$B$}
\psfrag{t}{}
\psfrag{1}{\Large$1$}
\psfrag{1+s}{\Large$1+s_N^{}$}
\includegraphics[angle=0, width=6cm, height=4.5cm]{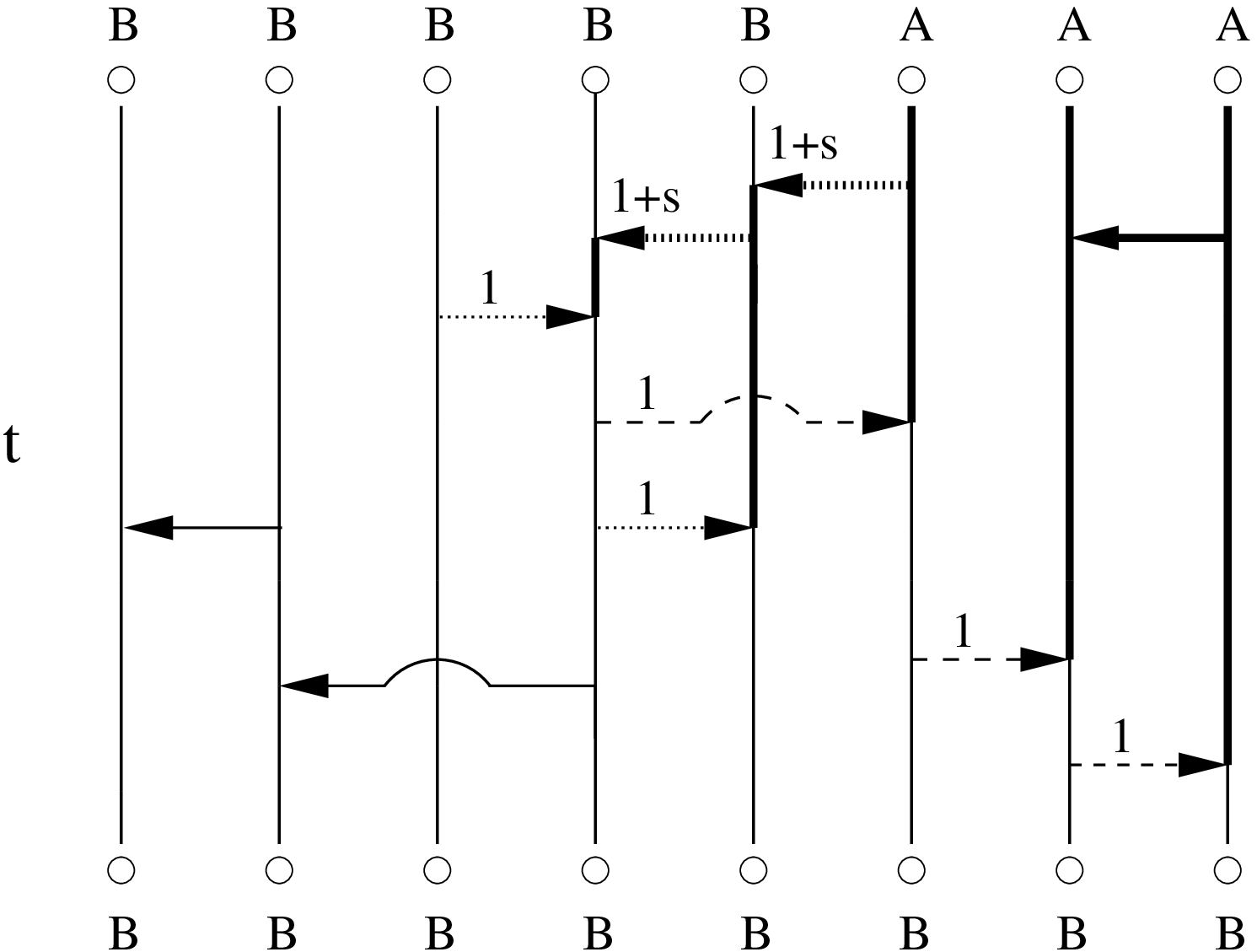}
\end{center}
\end{minipage}
\caption{Reflection principle in the (reduced) Moran model: $N=8$, $k=5$, realisations $\omega$ (left) and $\bar{\omega}$
(right). 
Altogether $P^8_{3}(\bar{\omega}) d \omega= (1+s_N^{})^{-3}  (1+s_N^{})^2 (1+s_N^{})^{-2} P^8_{5}(\omega) d \omega  $. }
\label{mirroring}
\end{center}
\end{figure}

To make the situation tractable, we will work with what we will call the
\textit{reduced Moran model}: Starting from the original Moran model, we remove those selective
arrows that appear between individuals that are both of 
type $A$. Obviously, this does not affect the process $Z^N_{}$
(in particular, it does not change the fixation probabilities), but it
changes the graphical representation. That is, in Fig. \ref{mirroring} unmarked arrows appear at rate $1/N$.

Let now $\Omega^N_{}$ be the set of graphical realisations of the reduced
Moran model for $t \in [0,T]$, where $T:=\min \{ T_0^N, T_N^N \}$. Let
$P^N_{k}$ be the probability measure on $\Omega^N_{}$, given $Z_0^N = k$, $1 \leq k \leq N-1$. For a realisation 
$\omega \in \Omega^N_{}$ with  $Z_0^N = k$, $1 \leq k \leq N-1$, in which $A$ becomes fixed (cf. Fig. \ref{mirroring}, left),
we define $\bar{\omega} \in \Omega^N_{}$
as the corresponding reflected realisation (cf. Fig. \ref{mirroring}, right).

Our strategy will now be to derive the fixation probabilities 
by comparing the weights of $\omega$ and $\bar{\omega}$.
To assess the relative
weights of $P^N_{k}(\omega) d\omega$ and $P^N_{N-k}(\bar{\omega}) d \omega$, note first that, since
$A$ becomes fixed in $\omega$, all 
$N - k$ type-$B$ individuals have to be replaced by type-$A$ individuals, i.e. by arrows that appear at rate
$(1+s_N^{})/N$. 
In $\bar{\omega}$, the corresponding arrows point from type-$B$ to type-$A$ individuals
and thus occur only at rate $1/N$. 
Second, we have to take into account so-called
\textit{external descendants} defined as follows:
\begin{definition}
 A descendant of a type-$i$ individual, $i \in S$, that originates by replacing an individual
 of a different type ($j \neq i$) 
 is termed external
 descendant of type $i$. 
 \label{new descendant}
\end{definition}
\begin{remark}
 The total number of external descendants of type $i$ is almost surely finite.
\end{remark}
Every external descendant of a type-$B$ individual goes back to an arrow that occurs at
rate $1/N$. If the type-$A$ individuals go to fixation (as in $\omega$), all external descendants of 
type $B$ must eventually be replaced by arrows that emanate from type-$A$ 
individuals and therefore occur at rate $(1 +s_N^{})/N$. 
In $\bar{\omega}$ this situation corresponds to external descendants of
type $A$ that originate from arrows at rate $(1 + s_N^{})/N$ and 
are eventually eliminated by arrows that emanate from type-$B$ individuals. 
In Fig. \ref{mirroring}, the births of external descendants
and their replacements are represented by dotted arrows, which always appear in pairs. 
Dashed arrows represent the
elimination of individuals (except external descendants) of the type that
eventually gets lost in the population. 

Let now $D^{N}_B(\omega)$ be the number of external descendants of type $B$ in $\omega$.
Then $D^{N}_B (\omega) < \infty$ almost surely and we obtain for the measure of
$\bar{\omega}$
\begin{align*}
 P^N_{N-k}(\bar{\omega}) d\omega  
 &= (1+s_N^{})^{-(N-k)} (1+s_N^{})^{D_B^{N}(\omega)} \Bigl( \frac{1}{1+s_N^{}} \Bigr)^{D_B^{N}(\omega)}  P^N_{k}(\omega) d\omega \displaybreak[0] \\
 &=(1+s_N^{})^{-(N-k)} P^N_{k}(\omega) d\omega.
\end{align*}
Note that the effects of creating external descendants and their replacement cancel 
each other, so that the relative weights of $P^N_{N-k}(\bar{\omega}) d\omega  $
and $ P^N_{k}(\omega) d\omega$ do not depend on $\omega$. Since reflection
provides a one-to-one correspondence between realisations with fixation of type $A$ and of type $B$,
respectively, we obtain the system of equations
\begin{equation}
1 - h_{N-k}^N =(1+s_N^{})^{-(N-k)}  h_k^N , \quad 1 \leq k \leq N-1 ,
\end{equation}
which is supplemented by $h_0^N = 0$ and $h_N^N =1 $, and which is 
solved by \eqref{h_i}.

\section{Labelled Moran model}
\label{Labelled Moran model}
In this Section we introduce a new particle model, which we call the
\textit{labelled Moran model}, and which has the same 
distribution of type frequencies as the original Moran model with selection of Sec. \ref{The Moran model with selection},
provided the initial conditions are chosen appropriately.
As before, we consider a population of fixed size $N$ in continuous time, but now 
every individual is assigned a label $i \in \{ 1, \dots , N \}$ with different
reproductive behaviour to be specified below. The initial population contains all $N$ labels
so that initially position $i$ in the graphical representation is occupied by label $i$, $1 \leq i \leq N$.
As in the original Moran model, birth events are represented by arrows; they lead to a single offspring, 
which 
inherits the parent's label, and replaces an individual as explained below. 
Again we distinguish
between neutral (at rate $1/N$) and selective (at rate $s_N^{}/N$) events.
Neutral arrows appear as before, at rate $1/N$ per ordered pair of lines, 
irrespective of their labels.
But we only allow for selective arrows emanating from a label $i$ and pointing
to a label $j > i$ (at rate $s_N^{}/N$ per ordered pair of lines with 
such labels). 
Equivalently, we may take together both types of arrows, such that an arrow points from a label 
$i$ to a label $j \leq i$ at rate $1/N$ and to label $j>i$ at rate $(1+s_N^{})/N$. We will
make use of both points of view.
The idea is to have
each label behave as `less fit' towards lower labels and as `more fit' towards higher labels.
 An example is given in Fig. \ref{labelmoran},
where neutral arrows are marked by $1$ and selective arrows by $s_N^{}$.

\begin{figure}[ht]
\begin{center}
\psfrag{t}{\Large$t$}
\psfrag{2}{\Large$2$}
\psfrag{1}{\Large$1$}
\psfrag{3}{\Large$3$}
\psfrag{4}{\Large$4$}
\psfrag{5}{\Large$5$}
\psfrag{6}{\Large$6$}
\psfrag{7}{\Large$7$}
\psfrag{8}{\Large$8$}
\psfrag{1+s}{\Large$1+s_N^{}$}
\psfrag{s}{\Large$s_N^{}$}
\includegraphics[angle=0, width=7cm, height=5cm]{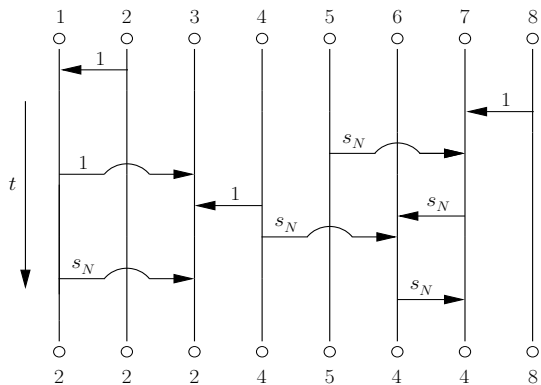}
\caption{The labelled Moran model. The labels are
indicated for the initial population (top) and a later one (bottom).
}
\label{labelmoran}
\end{center}
\end{figure}

\subsection{Ancestors and fixation probabilities}
Since there is no mutation in the labelled Moran model, one of the $N $ labels will eventually
take over in the population, i.e. one label will become fixed. We denote this label by $I^N_{}$ 
and term it the ancestor. Its distribution is given in
Thm. \ref{Thm Fixation}.

\begin{theorem}
 $I^N_{}$
 is distributed according to
 \begin{align} \label{eta_i Thm}
  \eta_i^{N}&:= \mathbb{P} (I^N_{}=i)
  = (1+s_N^{})^{N-i}_{} \eta_N^{N} = h_i^N - h_{i-1}^N, \ 1 \leq i \leq N, \\ 
  \intertext{with}
  \eta_N^{N}&:= \mathbb{P} (I^N_{}=N)
  = \frac{1}{\sum_{j=0}^{N-1}(1+s_N^{})^j_{}} = 1 - h^N_{N-1}.
  \label{eta_N}
 \end{align}
\label{Thm Fixation}
\end{theorem}
We will give two proofs of Thm. \ref{Thm Fixation}. The first provides an intuitive explanation
and is based on the type frequencies.
In the second proof, we will use an alternative approach in the spirit of the reflection
principle of Sec. \ref{Reflection principle}, which provides more insight into
the particle representation. This proof is somewhat more complicated, but it permits us to classify
reproduction events into those that have an effect on the fixation probability
of a given label and those that do not; this will become important later on. 
In analogy with Def. \ref{new descendant} we define \textit{external descendants of 
labels in $\mathcal{S}$}, $\mathcal{S}\subseteq \{1, \dots, N \}$, as descendants
of individuals with labels in $\mathcal{S}$ that originate by replacing an individual
of a label in the complement of $\mathcal{S}$.

\begin{proof}[First proof of Thm. \ref{Thm Fixation}]
For a given $i$, let $\bar{Z}^{N,i}_t$ be
the number of individuals with labels in $ \{ 1 , \dots , i \} $ at time $t$. 
Like $Z^N_{}$, the process $\bar{Z}^{N,i}_{}=\big( \bar{Z}^{N,i}_t \big)_{t \geq 0}$
is a birth-death process with rates $\lambda_j^N$ and $\mu_j^N$ of \eqref{eq:lambda_mu}.
This is because every individual with label in $ \{ 1 , \dots , i \} $
sends arrows into the set with labels in $ \{ i+1 , \dots , N \} $ at rate
$(1+s_N^{})/N$; in the opposite direction, the rate is $1/N$ per pair of individuals;
and arrows within the label classes do not matter. Since $\bar{Z}_0^{N,i} = i$,
the processes $\bar{Z}^{N,i}_{}$ and $Z^N_{}$ thus have the same law provided $Z_0^N = i$. 
As a consequence, 
$\mathbb{P} (I^N_{}\leq i) = \sum_{j=1}^i \eta_j^N = h_i^N$, $1 \leq i \leq N$, which,
together with \eqref{h_i}, immediately yields the assertions 
of Thm. \ref{Thm Fixation}. 
\end{proof}

\begin{proof}[Second proof of Thm. \ref{Thm Fixation}]
This proof aims at a direct calculation of $\eta_i^N$, $1 \leq i \leq N-1$,
as a function of $\eta_N^N$.
 Let a realisation of the labelled Moran model be given in which 
 label $N$ becomes fixed (this happens with probability $\eta_N^N$, still
 to be calculated).
 The basic idea now is to move every arrow 
 by way of a cyclic permutation of the arrows' positions, while keeping the initial ordering of the labels.
 More precisely, in the graphical representation
we move every arrow $i$ positions to the right (or, equivalently, $N-i$ positions to the left).
 That is, we shift an arrow that appears at time $t$ with tail at position $k$ and head at position $\ell$,
 such that it becomes an arrow that emanates from position $(k + i) \mod N$ and
 points to position $(\ell +i) \mod N$, again at time $t$.
 We thus obtain what we will call the \textit{permuted realisation of order $i$}. In
 this realisation, label $i$ is fixed,
 as illustrated in Fig. \ref{cyclic} for a labelled Moran model of size $N=8$. Here, the
 descendants of the label that becomes fixed are marked bold. Arrows that are marked 
 by $1$ respectively $1+s_N^{}$ appear at rate $1/N$ respectively
$(1+s_N^{}) /N$. A shift of every arrow
of $i = 5$ positions to the right transforms the left realisation,
in which label $N=8$ becomes fixed, into the right one (with fixation
of label $i=5$).
 
 Throughout, we keep the original meaning of the labels: Between every ordered pair of lines with
 labels $(i,j)$, arrows appear at rate $(1+s_N^{})/N$ if $j >i$, and at rate $1/N$
 otherwise. 
 Since, in the permuted realisation, we change the \textit{position} of each arrow, it now may
 affect a different pair of \textit{labels}, which may change the arrow's rate. As a result, the 
 permuted realisation has a different weight than the original one; this will now 
 be used to calculate $\eta_i^N$. (Mathematically, the following may be seen
 as a \textit{coupling argument}.)
 
 So, let $\Omega^N_{}$ be the set of realisations of the labelled Moran model 
 for $t \in [0,T]$, where $T$ now is the time at which
 one of the labels is fixed. Let $P^N_{}$ be the probability measure on $\Omega^N_{}$.
 For a realisation $\omega_N^{} \in \Omega^N_{}$ in which label $N$ becomes fixed (cf. Fig. \ref{cyclic}, left), we define
 $\omega_i^{} \in \Omega^N_{}$ as the corresponding permuted realisation of order $i$ (cf. Fig. \ref{cyclic}, right).
 We will now
 calculate $\eta_i^N$ by assessing the weight of the measure of $\omega_i^{}$ relative
 to that of $\omega_N^{}$. 
Below we run briefly through the cases to analyse the change of weight on the various types of arrows.
To this end, the label sets $\{ 1, \dots , N-i \}$ and $\{ N-i+1, \dots , N \}$ (left),
and $\{ 1, \dots , i \}$ and $\{ i+1, \dots , N \}$ (right), respectively,
are encircled at the top of Fig. \ref{cyclic}.\begin{figure}[ht]
\begin{center}
\begin{minipage}{0.4 \textwidth}
\begin{center}
\psfrag{1}{\Large$1$}
\psfrag{2}{\Large$2$}
\psfrag{3}{\Large$3$}
\psfrag{4}{\Large$4$}
\psfrag{5}{\Large$5$}
\psfrag{6}{\Large$6$}
\psfrag{7}{\Large$7$}
\psfrag{8}{\Large$8$}
\psfrag{1+s}{\Large$1+s_N^{}$}
\includegraphics[angle=0, width=6cm, height=4.5cm]{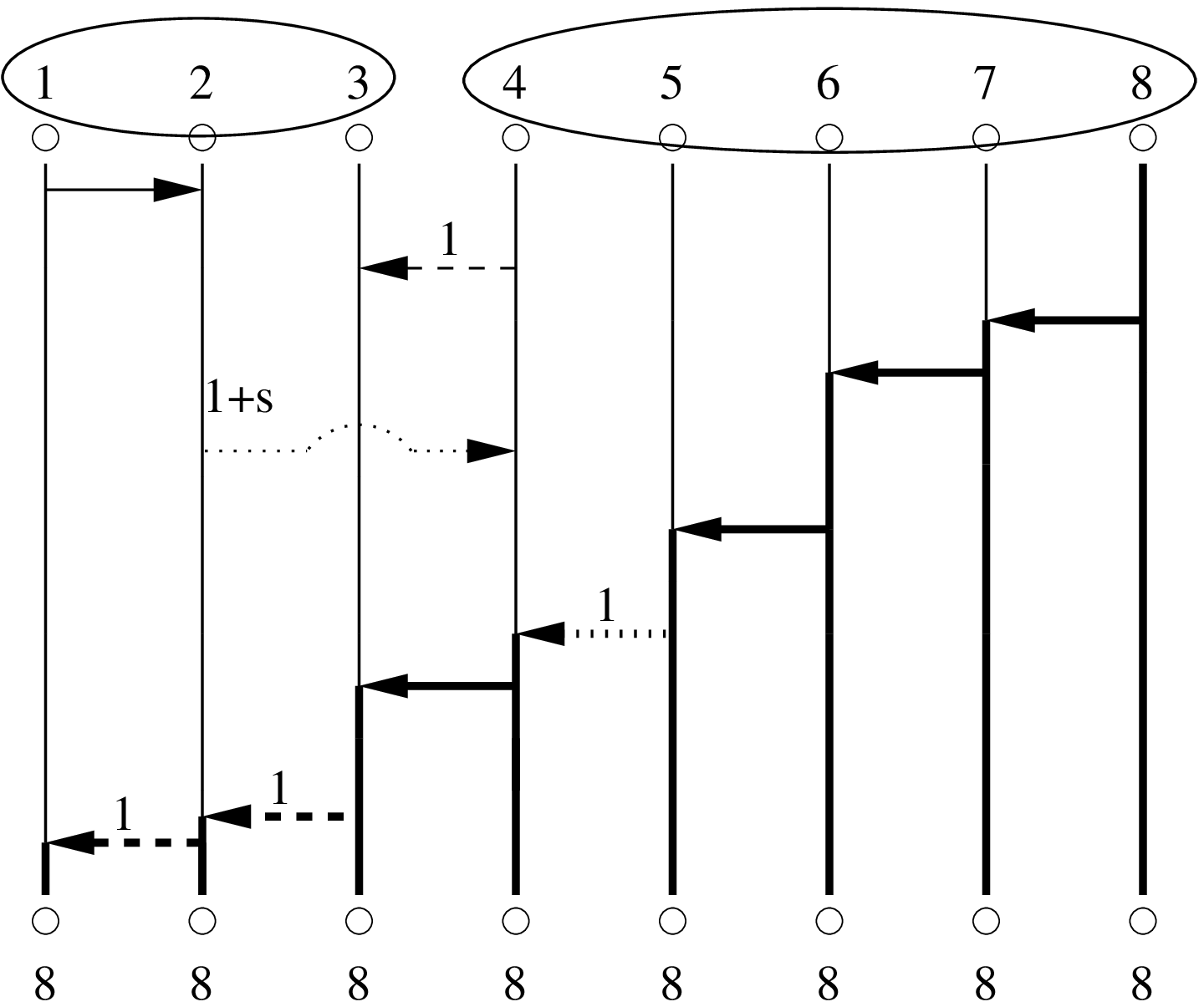}
\end{center}
\end{minipage}
\hspace{0.5cm}
\begin{minipage}{0.4 \textwidth}
\begin{center}
\psfrag{1}{\Large$1$}
\psfrag{2}{\Large$2$}
\psfrag{3}{\Large$3$}
\psfrag{4}{\Large$4$}
\psfrag{5}{\Large$5$}
\psfrag{6}{\Large$6$}
\psfrag{7}{\Large$7$}
\psfrag{8}{\Large$8$}
\psfrag{1+s}{\Large$1+s_N^{}$}
\includegraphics[angle=0, width=6cm, height=4.5cm]{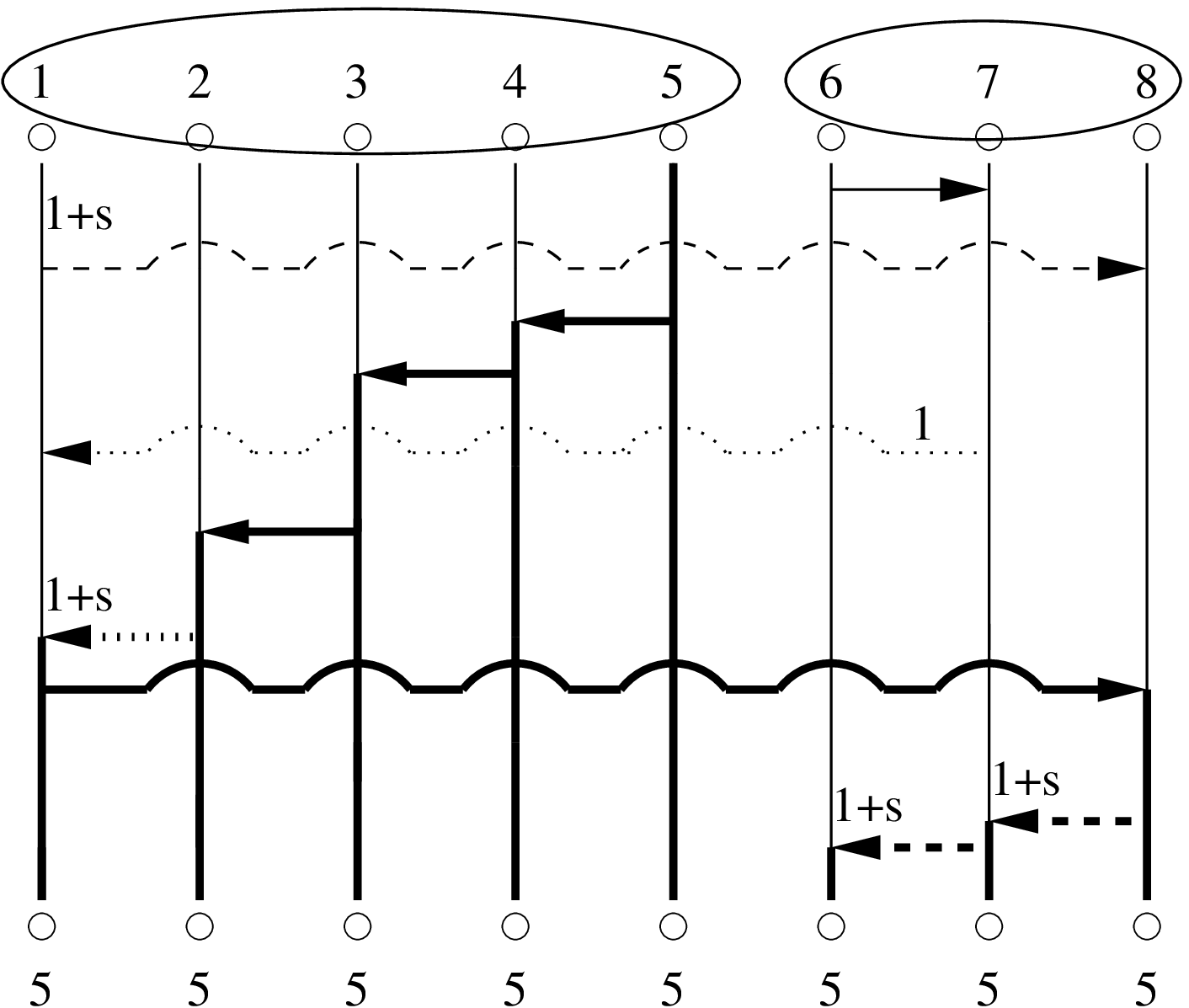}
\end{center}
\end{minipage}
\caption{Cyclic permutation in the labelled Moran model: $N=8$, $i=5$, realisations $\omega^{}_8$ (left) and
$\omega^{}_5$ (right).
Altogether $P^8_{} (\omega_5^{} ) d \omega_8^{}= (1+s_N^{})^{-1} (1+s_N^{})(1+s_N^{})^3 P^8_{} (\omega_8^{} ) d \omega_8^{}$.}
\label{cyclic}
\end{center}
\end{figure}

 \begin{enumerate}
\item[(a)] Arrows in $\omega^{}_N$ that point from and to labels within the 
sets $\{ 1, \dots , N-i \}$ or $\{ N-i+1, \dots , N \}$, respectively, 
turn into arrows within the sets $\{i+ 1, \dots , N \}$
or $\{ 1, \dots , i \}$, respectively, under the permutation. Such arrows retain their
`direction' (with respect to the labels) and thus appear at identical
rates in $\omega^{}_N$ and $\omega^{}_i$ (cf. solid arrows in Fig. \ref{cyclic}). 
 
 \item[(b)] Arrows in $\omega^{}_N$ that emanate from
 the set of labels $\{1, \dots , N-i \}$ and point to the set of labels 
 $\{N-i+1, \dots , N \}$ occur at rate $(1 +s_N^{})/N$ and create
 external descendants of labels in $\{1, \dots , N-i \}$.
 Since label $N$ becomes fixed, every such external descendant 
 is eventually eliminated by an arrow at rate $1/N$.
 The corresponding situation in $\omega^{}_i$ 
 concerns external descendants of labels in $\{i+ 1, \dots , N \}$, which result
 from neutral arrows and finally are replaced at rate $(1 + s_N^{})/N$ each.
 In Fig. \ref{cyclic} there is exactly one external descendant of the labels in $\{ 1, \dots , N-i \}$ (left)
and $\{ i+ 1, \dots , N \}$ (right), respectively. It originates from the respective
first dotted arrow, and is replaced via the second dotted arrow.
 
 \item[(c)] It remains to deal with the replacement of the labels $1, \dots , N-i $
 (except for their external descendants) in $\omega^{}_N$.
 Exactly $N-i$ neutral arrows are responsible for this, they transform 
 into arrows at rate $(1 + s_N^{})/N$ through the permutation. 
 These are the dashed arrows in Fig. \ref{cyclic}. 
 \end{enumerate}
 
 Let now $D_{\leq i}^{N} (\omega_N^{})$ be the number of external descendants of labels in $\{1, \dots , i \}$
 in $\omega^{}_N$ ($D_{\leq i}^{N} (\omega_N^{}) < \infty$ almost surely).
 Then, (a) - (c) yield for the measure of $\omega^{}_i$
 \begin{align*}
  P^N_{} (\omega_i^{} ) d \omega_N^{}
  &=(1+s_N^{})^{D_{\leq i}^{N}(\omega_N^{})} \Bigl( \frac{1}{1 + s_N^{}} \Bigr)^{D_{\leq i}^{N}(\omega_N^{})}  (1+s_N^{})^{N-i} P^N_{} (\omega_N^{}) d \omega_N^{} \displaybreak[0] \\
  &= (1+s_N^{})^{N-i} P^N_{} (\omega_N^{}) d \omega_N^{} .
 \end{align*}
As in the reflection principle, the effects of the external descendants cancel each other and so the 
relative weights of $ P^N_{} (\omega_i^{}) d \omega_N^{}$ and $P^N_{} (\omega_N^{}) d \omega_N^{}$
are independent of the particular choice of $\omega^{}_N$. Since the cyclic permutation yields 
a one-to-one correspondence between realisations that lead to fixation of label $N$ and 
label $i$, respectively, we obtain
\begin{equation}
 \eta_i^N = (1+s_N^{})^{N-i} \eta_N^N, \quad 1 \leq i \leq N .
 \label{eta_i Beweis}
\end{equation}
Finally, the normalisation
\begin{equation*}
 1 = \sum_{i=1}^N \eta_i^N = \eta_N^N \sum_{i=1}^N (1+s_N^{})^{N-i}
\end{equation*}
yields the assertion of Thm. \ref{Thm Fixation}.
 
\end{proof}

\subsection{Defining events}
\label{Defining events}
We are now ready to investigate the effect of selection in more depth.
The second proof of Thm. \ref{Thm Fixation} 
allows us to identify the reproduction
events that affect the distribution of $I^N_{}$ (i.e. that are responsible
for the factor $(1+s_N^{})^{N-I_{}^N}$ in \eqref{eta_i Thm}) with the dashed
arrows in Fig. \ref{cyclic} (cf. case (c)), whereas dotted arrows appear pairwise and their effects cancel
each other (cf. case (b)). It is natural to term these reproductions \textit{defining events}:
\begin{definition}
A defining
 event is an arrow that emanates from the set of labels $\{1, \dots , I^N_{} \}$
 and targets individuals with labels in the set $\{ I^N_{}+1, \dots , N \}$ that are not 
 external descendants of labels in $\{ I^N_{}+1, \dots , N \}$.
\end{definition}
Loosely speaking, a defining event occurs every time the descendants of $\{1, \dots , I^N_{}\}$
`advance to the right'.
In particular, for every $j \in \{ I^N_{}+1, \dots , N \}$, the first arrow that emanates from 
a label in $\{ 1, \dots , I^N_{}  \}$ and hits the individual at position $j$ is a defining event.
Altogether, there will be $N - I^N_{}$ defining events until fixation.
It is important to note that they need not be reproduction
events of the ancestral label $I^N_{}$ itself. 

See Fig. \ref{defining} for an illustration, in which the ancestor $I^N_{}$ is indicated at the top, its descendants  are 
marked bold, and the defining events are represented by dashed arrows.
On the left, both defining events are reproduction events of $I^N_{}$. The
arrows that are indicated by \textasteriskcentered \
and \textasteriskcentered\textasteriskcentered \ 
are not defining events, since the first one (\textasteriskcentered) concerns
only labels within $\{ I^N_{}+1, \dots , N \}$ and the second one 
(\textasteriskcentered\textasteriskcentered) targets an external descendant of 
$\{ I^N_{}+1, \dots , N \}$. On the
right, only the second defining event is a reproduction event of 
$I^N_{}$, whereas the first one in a reproduction event of a label less than $I^N_{}$. \begin{figure}[ht]
\begin{center}
\begin{minipage}{0.3 \textwidth}
\begin{center}
\psfrag{P}{\Large$I^N_{}$}
\psfrag{*}{\Large\textasteriskcentered}
\psfrag{**}{\Large\textasteriskcentered\textasteriskcentered}
\includegraphics[angle=0, width=3.5cm, height=4.2cm]{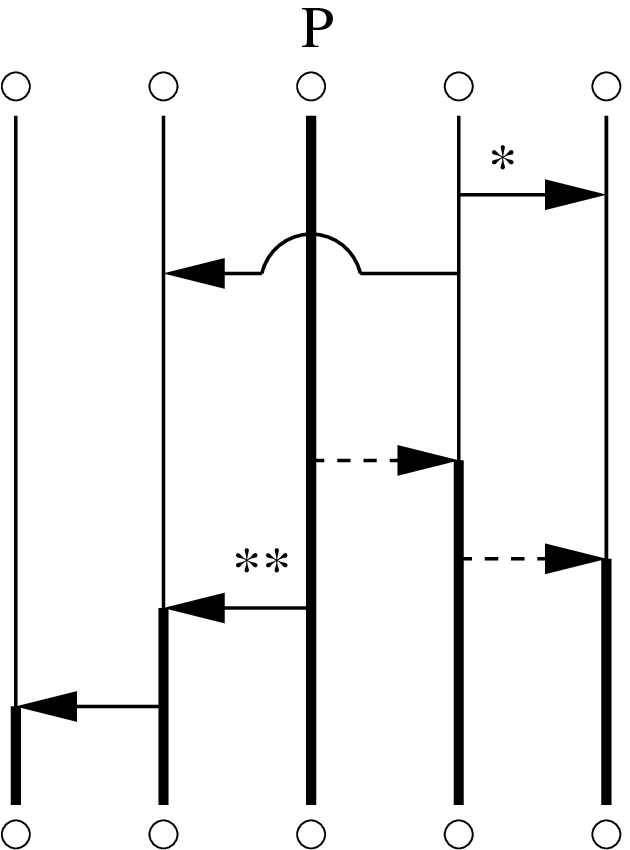}
\end{center}
\end{minipage}
\hspace{0.5cm}
\begin{minipage}{0.3 \textwidth}
\begin{center}
\psfrag{P}{\Large$I^N_{}$}
\includegraphics[angle=0, width=3.5cm, height=4.2cm]{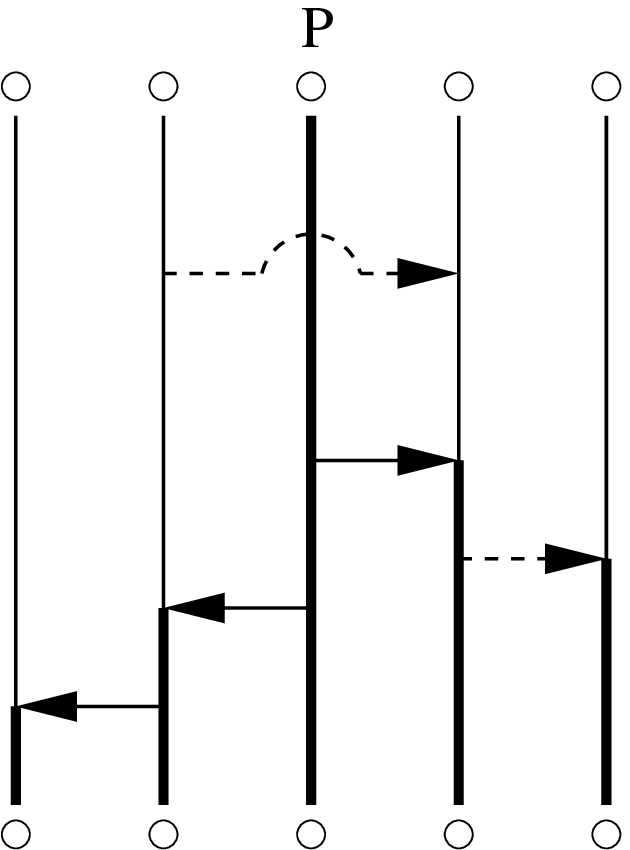}
\end{center}
\end{minipage}
\caption{Defining events in the labelled Moran model.}
\label{defining}
\end{center}
\end{figure}

It is clear that all defining events appear at rate $(1+s_N^{})/N$. Decomposing these arrows into 
neutral and selective ones reveals 
that
each defining event is either a selective (probability 
$s_N^{}/(1+s_N^{})$) or a neutral (probability $1/(1+s_N^{})$) reproduction event, 
independently of the other defining events and of $I^N_{}$. Let $V_i^N$, $I^N_{}+1 \leq i \leq N$, be the corresponding family
of Bernoulli random variables that indicate whether the respective defining event
is selective. Let $Y^N_{}$ denote the 
\textit{number of selective defining events}, that is, 
\begin{equation}
Y^N_{} := \sum_{i=I^N_{}+1}^{N} V_i^N 
\label{Y^N}
\end{equation}
with the independent and identically distributed (i.i.d.) Bernoulli variables $V_i^N$
just defined. $Y^N_{}$ will turn out as a pivotal quantity for everything to
follow. We now characterise its distribution, and the dependence between $I^N_{}$
and $Y^N_{}$. This will also provide us with an interpretation of the coefficients
$a_n^N$ in \eqref{h_i_alt}, as well as another  representation of the 
fixation probabilities.

It is clear from \eqref{Y^N} that $Y^N_{}$ is distributed according to Bin$(N - I^N_{}, s_N^{}/(1+s_N^{}) )$, 
the binomial distribution with parameters
$N - I^N_{}$ and $s_N^{}/(1+s_N^{})$ in the
sense of a two-stage random experiment. That is, given $I^{N}_{}=i$, $Y^{N}_{}$ follows
Bin$(N-i,s_N^{}/(1+s_N^{}) )$.  
Thus, for $0 \leq n \leq N-i$, we obtain (via
\eqref{eta_i Thm}) 
\begin{align}
 &\mathbb{P}(Y^N_{}=n, I^N_{} = i) =
 \binom{N-i}{n}  \Bigl( \frac{s_N^{}}{1+s_N^{}} \Bigr)^n \Bigl( \frac{1}{1+s_N^{}} \Bigr)^{N-i-n} \eta_i^N
 = \binom{N-i}{n} s^n_N \eta_N^N 
 \label{verteilung selective def. events, fixierende Linie}\\
 \intertext{and thus}
& \mathbb{P}(Y^N_{}=n) = \sum_{i=1}^{N-n} \binom{N-i}{n} s^n_{N} \eta_N^N
=\sum_{i=n}^{N-1} \binom{i}{n} s^n_{N} \eta_N^N
 = \binom{N}{n+1} s^n_{N} \eta_N^N,
 \label{verteilung selective def. events}
\end{align}
where the last equality is due to the well-known identity 
\begin{equation}
 \sum_{i= \ell }^k \binom{i}{\ell} = \binom{k+1}{\ell +1},  \quad \ell, k \in \mathbb{N}_0, 0 \leq \ell \leq k.
 \label{identity binom}
\end{equation}
In particular, $\mathbb{P}(Y^N_{}=0)= N \eta_N^N$.
Obviously, $\mathbb{P}(Y^N_{}=n)$, $n \geq 1$, may also be expressed recursively as
\begin{equation}
 \mathbb{P}(Y^N_{}=n)= s_N^{} \frac{N-n}{n+1}  \mathbb{P} (Y^N_{}=n-1).
 \label{verteilung selective def events neu2}
\end{equation}
Furthermore, equations \eqref{verteilung selective def. events, fixierende Linie} and
\eqref{verteilung selective def. events} immediately yield the
conditional probability
\begin{equation}
 \mathbb{P} (I^N_{} =i \mid Y^N_{} = n) = \frac{\binom{N-i}{n}}{\binom{N}{n+1}}.
 \label{conditional fixation}
\end{equation}

Compare now \eqref{verteilung selective def events neu2} 
with \eqref{a_n} and, for the initial condition, compare \eqref{verteilung selective def. events}
with \eqref{a_0} (with the help of \eqref{eta_N}). Obviously, the  $\mathbb{P}(Y^N_{} = n)$
and the $a_n^N-a_{n+1}^N$, $0 \leq n \leq N-2$, have the same initial value (at $n=0$) and follow the same recursion,
so they agree. Via normalisation, we further have
\begin{equation*}
\mathbb{P} (Y^N_{} = N-1) = 1 - \sum_{n=0}^{N-2} (a_n^N -a^N_{n+1}) = a_{N-1}^N.
\end{equation*}
We have therefore found a probabilistic meaning for the $a_n^N$, $0 \leq n \leq N-1$,
namely, 
\begin{equation}
a_n^N = \mathbb{P}(Y^N_{} \geq n)  .
\label{a_n^N Interpretation}
\end{equation}
This will
be crucial in what follows.

Another interesting characterisation is the following:
\begin{proposition}
Let $W^N_{} := Y^N_{} + 1$. Then $W^N_{}$ has the distribution of a random variable that follows 
\rm{Bin$(N, s_N^{}/(1+s_N^{}))$}, \textit{conditioned to be positive.}
\label{Proposition Y+1}
\end{proposition}

\begin{proof}
 Let $A$ be a random variable distributed according to Bin$(N, s_N^{}/(1+s_N^{}))$. Then
 $$
 \mathbb{P}(A > 0) = 1- \frac{1}{(1+s_N^{})^N}
 = \Big( 1 - \frac{1}{1+s_N^{}} \Big) \sum_{i=0}^{N-1} \frac{1}{(1+s_N^{})^i}
 $$
 by the geometric series. For $n \geq 1$, therefore, 
 $$
 \mathbb{P}(A=n \mid A > 0)
 = \frac{\binom{N}{n} s^n_N(1+s_N^{})^{-N} }{\frac{s_N^{}}{1+s_N^{}} \sum_{i=0}^{N-1} (1+s_N^{})^{-i}  }
  = \frac{\binom{N}{n} s^{n-1}_N }{ \sum_{i=0}^{N-1} (1+s_N^{})^{i}  }
  = \mathbb{P}(Y^N=n-1),
 $$
 where the last step is due to \eqref{verteilung selective def. events}.
 This proves the claim. 
\end{proof}

Note that label $N$ is obviously not capable of selective reproduction 
events and its fixation implies the absence of (selective) defining events.
Its fixation probability $\eta_N^N$ coincides with the fixation probability
of any label $i$, $1 \leq i \leq N$, in the absence of any selective defining
events: $\mathbb{P}(Y^N_{}=0, I^N_{} = i) = \eta_N^N$, cf. \eqref{verteilung selective def. events, fixierende Linie}.
For this reason, we term $\eta_N^N$ the \textit{basic fixation probability}
of every label $i$, $1 \leq i \leq N$, and express all relevant quantities  
in terms of $\eta_N^N$. Note that, for $s_N^{} >0$, $\eta_N^N$ is different from the \textit{neutral
fixation probability}, $\eta_i^N = 1/N$, that applies to every label $i$, $1 \leq i \leq N$, in the case
$s_N^{}=0$. \\

Decomposition according to the number of selective defining events yields
a further alternative representation of the fixation probability $h_i^N$ (cf. \eqref{h_i}) of the
Moran model. Consider
\begin{equation}
\begin{split}
 \mathbb{P}(I^N \leq i \mid Y^N_{} =n) &= 
 \frac{1}{\binom{N}{n+1}}\sum_{j=1}^i \binom{N-j}{n} 
 =  \frac{1}{\binom{N}{n+1}} \sum_{j=N-i}^{N-1} \binom{j}{n}  
 = \frac{\binom{N}{n+1}-\binom{N-i}{n+1} }{\binom{N}{n+1}}, 
 \label{gemeinsame Vtlg Y I}
 \end{split}
\end{equation}
where we have used \eqref{conditional fixation} and \eqref{identity binom}.
This leads us to the following series expansion in $s_N^{}$: 
\begin{equation}
\begin{split}
 h_i^N = \mathbb{P}(I^N \leq i) 
 &= \sum_{n=0}^{N-1} \mathbb{P}(I^N \leq i \mid Y^N_{} =n) \mathbb{P} (Y^N_{} = n) \\
 &= \sum_{n=0}^{N-1} \bigg[ \binom{N}{n+1}-\binom{N-i}{n+1} \bigg]  s^n_N \eta_N^N.
 \end{split}
 \label{h_i alternativ}
\end{equation}
We will come back to this in the next Section.

\section{$Y^N_{}$ in the diffusion limit}
\label{Diffusion limit of the labelled Moran model}
In this Section we analyse the number of selective defining events in the diffusion
limit (cf. Sec. \ref{The Moran model with selection}).
First of all we recapitulate from \eqref{verteilung selective def. events, fixierende Linie}
and \eqref{eta_i Thm} that
\begin{equation*}
\mathbb{P}(Y^N_{}=n, I^N_{} \leq i) = 
\frac{1}{N} \sum_{j=1}^{i}  \binom{N-j}{n}  \Bigl( \frac{s_N^{}}{1+s_N^{}} \Bigr)^n \Bigl( \frac{1}{1+s_N^{}} \Bigr)^{N-j-n}
N (h_j^N - h_{j-1}^N).
\end{equation*}
For a sequence $(i_N^{})_{N \in \mathbb{N}}$ with $i_N \in \{ 1, \dots , N \}$, 
$\lim_{N \to \infty} i_N^{} / N=x$, and $x \in [0,1]$, 
this yields
\begin{equation}
 \lim_{N \to \infty} \mathbb{P}(Y^N_{}=n, I^N_{}/N \leq i_N^{}/N) = 
 \int_0^x \frac{(\sigma(1-y))^n_{}}{n!} \exp (- \sigma (1-y)) h'(y) dy, 
 \label{gemeinsameVerteilungLimes}
\end{equation}
where we have used the convergence of the binomial to the Poisson distribution. Thus, the 
sequence of random variables $(Y^N_{}, I^N_{}/N)_{N \in \mathbb{N}}$ converges in distribution
to a pair $(Y^{\infty}_{}, I^{\infty}_{})$ of random
variables with values in $\mathbb{N}_0 \times [0,1]$ and 
distribution function \eqref{gemeinsameVerteilungLimes}.
Marginalisation with respect to the second variable implies that
$I^{\infty}_{}$
has distribution function $h$. 
As a consequence, $Y^{\infty}_{}$ follows Poi$(\sigma (1-I^{\infty}_{}))$, the Poisson distribution 
with parameter $\sigma (1-I^{\infty}_{})$.

Since
\begin{equation}
 \mathbb{P}(Y^{\infty}_{} \geq n) = \lim_{N \to \infty} \mathbb{P}(Y^N_{}\geq n)
 = \lim_{N \to \infty} a_n^N = a_n^{},
 \label{a_n Interpretation}
\end{equation}
with coefficients $a_n^{}$ as in \eqref{a_n Fearnhead} and \eqref{a_0 Fearnhead},
we may conclude that $\mathbb{P} (Y^{\infty}_{}=n) = a_n^{} - a_{n+1}^{}$ for $n \geq 0$. 
We thus have found an interpretation of the coefficients $a_n^{}$ in \eqref{h(x)}. In particular,
\begin{equation}
 \begin{split}
 \mathbb{P}(Y^{\infty}_{} =0) &= \lim_{N \to \infty} \mathbb{P}(Y^N_{}=0)
 =\lim_{N \to \infty} N \eta_N^N  
 =\lim_{N \to \infty} \left[ \frac{1}{N} \sum_{j=0}^{N-1} \Bigl(1+\frac{Ns_N^{}}{N} \Bigr)^{\frac{j}{N}N}  \right]^{-1} \\
 &= \left[ \int_{0}^{1} \exp(\sigma p) dp \right]^{-1}
 = \frac{\sigma}{\exp(\sigma)-1},
 \label{a1 label}
 \end{split}
\end{equation}
where we have used \eqref{verteilung selective def. events} and \eqref{eta_N}.
According to \eqref{a_n Fearnhead} we have the recursion
\begin{equation} \label{an recursion}
 \mathbb{P}(Y^{\infty}_{}=n)
 =\frac{\sigma}{n+1}\mathbb{P}(Y^{\infty}_{}=n-1) 
 \end{equation}
for $n \geq 1$ and iteratively via \eqref{an recursion} and \eqref{a1 label} 
\begin{equation}
  \mathbb{P}(Y^{\infty}_{}=n)=\frac{\sigma^n_{}}{(n+1)!}\mathbb{P}(Y^{\infty}_{}=0)
 =\frac{\sigma^{n+1}_{}}{(n+1)!(\exp(\sigma)-1)} \label{an recursion2}
\end{equation}
for $n \geq 0$. With an argument analogous to that in Prop. \ref{Proposition Y+1}, one also obtains
that $W^{\infty}_{} := Y^{\infty}_{} + 1$ has the distribution of a random variable following 
Poi$(\sigma)$, conditioned to be positive.

We now aim at expressing $h$ in terms of a decomposition according to the
values of $Y^{\infty}_{}$ (in analogy with \eqref{h_i alternativ}). 
We recapitulate equation \eqref{gemeinsame Vtlg Y I} 
to obtain
\begin{equation*}
\mathbb{P} (I^{\infty}_{} \leq x \mid Y^{\infty}_{}=n)
 =\lim_{N \to \infty} \mathbb{P}(I^N_{} \leq i_N^{} \mid Y^N_{}=n)
= \lim_{N \to \infty} 1 - \frac{\binom{N-i_N^{}}{n+1}}{\binom{N}{n+1}} = 1-(1-x)^{n+1}.
\end{equation*}
Then, the equivalent to \eqref{h_i alternativ} is a series expansion in $\sigma$:
\begin{equation}
 \begin{split}
 h(x)&= \mathbb{P} (I^{\infty}_{} \leq x)
 = \sum_{n \geq 0} \mathbb{P} (I^{\infty}_{} \leq x \mid Y^{\infty}_{}=n)\mathbb{P} ( Y^{\infty}_{}=n)  \\
 &= \frac{1}{\exp(\sigma) - 1} \sum_{n \geq 1} \frac{1}{n!} (1-(1-x)^n) \sigma^n_{}.
 \label{h diffusion}
 \end{split}
\end{equation}
Note that this can also be derived directly from \eqref{h} by a Taylor expansion around
the point $x=1$.


Finally note that \eqref{h diffusion} may also be expressed as
\begin{equation}
 h(x)= \sum_{n \geq 1} (1-(1-x)^{n}) (a_{n-1}^{} - a_{n}^{})
 = \mathbb{E} (1-(1-x)^{W^{\infty}_{}}) .
 \label{h Pokalyuk}
\end{equation}
Interestingly, but not unsurprisingly, this coincides with a representation given by Pokalyuk and Pfaffelhuber
(2013, Lemma 2.2)
in their proof of Kimura's fixation probability \eqref{h}.
They follow an argument of Mano (2009) that establishes a connection between the ASG at stationarity and the
fixation probability. A key here is the insight that the number of lines in the ASG at stationarity has the distribution
of a random variable that follows
Poi$(\sigma)$, conditioned to be positive -- which coincides with the distribution of $W^{\infty}_{}$.

\section{The targets of selective defining events and construction of the ancestral line}
\label{A simulation algorithm}
We have, so far, been concerned with the ancestral label and with the number 
of selective defining events, but have not investigated the \textit{targets}
of these events yet. This will now be done. 
We begin with another definition.
\begin{definition}
 Let $Y^N_{}$ take the value $n$. We then denote by $J_1^N, \dots, J_n^N$, with $J_1^N < \dots < J_n^N$,
 the (random) positions that are hit by the $n$ selective defining events.
\end{definition}
See Fig. \ref{targets} for an example, in which $N=5$, $I^N_{}=2$, $Y^N_{}=2$, $J_1^N=3$ and $J_2^N=5$. Note that the 
first selective defining event hits position $3$, which is occupied by an individual
of label $4$ at that time.
\begin{figure}[ht]
\begin{center}
\psfrag{P}{\Large$I^N_{}$}
\psfrag{s}{\Large$s_N^{}$}
\psfrag{1}{\Large$1$}
\includegraphics[angle=0, width=3.5cm, height=4.2cm]{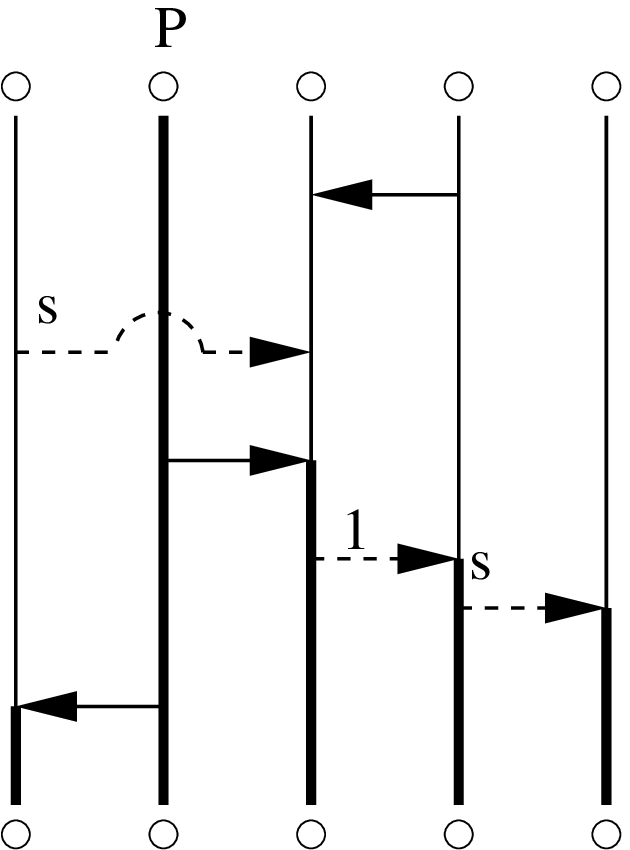}
\caption{Targets of selective defining events.} 
\label{targets}
\end{center}
\end{figure} 

In terms of the family $V_i^{N}$, $I^N_{}+1 \leq i \leq N$, of i.i.d. 
Bernoulli random variables (cf. Sec. \ref{Defining events}), the $J_1^N, \dots, J_{Y^N_{}}^N$ 
are characterised as
\begin{equation*}
\{J_1^N, \dots, J_{Y^N_{}}^N \} = \{ i \in \{ I^N_{} +1, \dots , N \} : V_i^N =1 \} .
\end{equation*}
(For $Y^N_{}=0$ we have the empty set.) Given that $Y^N_{}=n$, the 
$n$-tuple $(J_1^N, \dots, J_n^N)$ is uniformly distributed (without replacement) on the set of positions
$\{ I^N_{}+1, \dots, N \}$:
\begin{equation*}
 \mathbb{P}(J^N_1 = j_1^{}, \dots , J^N_n= j_n^{} \mid I^N_{}=i, Y^N_{}=n)
 = \frac{1}{\binom{N-i}{n}},
\end{equation*}
which implies (via \eqref{conditional fixation})
\begin{equation}
 \mathbb{P}(I^N_{}=i, J^N_1 = j_1^{}, \dots , J^N_n= j_n^{} \mid  Y^N_{}=n)
 = \frac{1}{\binom{N}{n+1}}.
 \label{verteilung tupel}
\end{equation}
Hence, the $(n+1)$-tuples $(I^N_{},J_1^N, \dots, J_n^N) $ are sampled from $\{1, \dots , N \}$
uniformly without replacement. Note that $W^N$ of Prop.~1 is the size of this
tuple. Note also that the tuples 
do not contain any further information
about the appearance of arrows in the particle picture. In particular, we do not learn which
label `sends' the arrows, nor do we include selective arrows that do not belong
to defining events.

 It is instructive to formulate
a simulation algorithm (or a construction rule) for these tuples. 
\begin{algorithm}
First draw the number of selective defining events, that is, a realisation $n$ of $Y^N_{}$ 
(according to \eqref{verteilung selective def. events}).
Then simulate $(I^N_{},J_1^N, \dots, J_n^N) $ in the following way: 
 \begin{description}
 \item [Step 0:] Generate a random number $U^{(0)}_{}$ that is uniformly distributed 
 on $\{1, \dots, N \}$. Define $\mathcal{I}^{(0)}_{} := U^{(0)}$. If 
 $n>0$ continue with step $1$, otherwise stop.
 \item[Step 1:] Generate (independently of $U^{(0)}_{}$) a second random number 
 $U^{(1)}_{}$ that is uniformly distributed on $\{1, \dots , N \} \setminus U^{(0)}_{}$.
 \begin{description}
  \item[(a)] If $U^{(1)}_{} > \mathcal{I}^{(0)}_{}$, define 
  $\mathcal{I}^{(1)}_{} := \mathcal{I}^{(0)}_{}$,
  $\mathcal{J}^{(1)}_{1} := U^{(1)}_{}$.
  \item[(b)] If $U^{(1)}_{} < \mathcal{I}^{(0)}_{}$, define $\mathcal{I}^{(1)}_{} :=  U^{(1)}_{}$,
  $\mathcal{J}^{(1)}_{1} :=\mathcal{I}^{(0)}_{}$.
 \end{description}
If $n>1$ continue with step $2$, otherwise stop.
\item[Step k:] Generate (independently of $U^{(0)}_{}, \dots, U^{(k-1)}_{}$) a random number 
 $U^{(k)}_{}$ that is uniformly distributed on $\{1, \dots , N \} \setminus \{ U^{(0)}_{}, \dots , U^{(k-1)}_{} \}$. 
 \begin{description}
  \item[(a)] If $U^{(k)}_{} > \mathcal{I}^{(k-1)}_{}$, define 
  $\mathcal{I}^{(k)}_{} := \mathcal{I}^{(k-1)}_{}$
  and assign the variables $U^{(k)}_{}, \mathcal{J}_1^{(k-1)}, \dots , \mathcal{J}_{k-1}^{(k-1)}$
  to $\mathcal{J}_1^{(k)}, \dots , \mathcal{J}_{k}^{(k)}$, such that 
  $\mathcal{J}_1^{(k)}< \dots < \mathcal{J}_{k}^{(k)}$.
  \item[(b)] If $U^{(k)}_{} < \mathcal{I}^{(k-1)}_{}$, define 
  $\mathcal{I}^{(k)}_{} :=  U^{(k)}_{}$,
  $\mathcal{J}^{(k)}_{1} :=\mathcal{I}^{(k-1)}_{}$ and 
  $\mathcal{J}^{(k)}_{\ell} :=\mathcal{J}^{(k-1)}_{\ell - 1}$
  for $2 \leq \ell \leq k$.
 \end{description}
If $n>k$ continue with step $k+1$, otherwise stop.
\end{description}
\label{algorithm}
\end{algorithm}
The simulation algorithm first produces a vector $(U^{(0)_{}}, \dots, U^{(n)}_{})$
uniformly distributed on the set of \textit{unordered} $(n+1)$-tuples and then turns it
into the vector $(\mathcal{I}^{(n)}_{}, \mathcal{J}^{(n)}_{1}, \dots , \mathcal{J}^{(n)}_{n})$,
which is uniformly distributed on the set of \textit{ordered} $(n+1)$-tuples. The latter therefore 
has the same distribution as the vector of random variables 
$(I^N_{}, J_1^N, \dots , J_n^{N})$, given $Y^N_{}=n$ (cf. \eqref{verteilung tupel}).
The interesting point now is that we may interpret
the algorithm as a procedure for the construction of the ancestral line. It successively
adds selective arrows to realisations of the labelled Moran
model, such that they coincide with additional selective defining events;
obviously, $n+1$ is the number of steps.

This genealogical interpretation is best explained with the help of the illustrations 
in Figs. \ref{step1} and \ref{stepk}. For each graphical representation the corresponding
ancestors ($\mathcal{I}^{(0)}_{}$, $\mathcal{I}^{(1)}_{}$, $\mathcal{I}^{(k-1)}_{}$,
and $\mathcal{I}^{(k)}_{}$,
respectively) and, if present, the targets of the selective defining events ($\mathcal{J}^{(1)}_{1}$,
$\mathcal{J}^{(k-1)}_{1}, \dots , \mathcal{J}^{(k-1)}_{k-1}$, and 
$\mathcal{J}^{(k)}_{1}, \dots , \mathcal{J}^{(k)}_{k}$, respectively) are indicated at
the top. Bold lines represent the genealogy of the 
entire population at the bottom. 
In step $0$ we randomly choose one of the $N$ labels. This represents
the label that becomes fixed, i.e. the ancestor, in a particle representation with no
selective defining events (cf. Fig. \ref{step1}, left). (Actually, this coincides with
the neutral situation, $s_N^{}=0$.) In the following steps selective arrows
are added one by one, where the point to note is that each of them may or may not move the ancestor
`to the left', depending on whether (a) or (b) applies. 
Lines that have been ancestors in previous steps of the algorithm are 
represented by dotted lines.
For instance, let us consider step $k$, that is, a realisation
with $k-1$ selective defining events is augmented by a further one.
In case (a) (cf. Figs. \ref{step1} and \ref{stepk}, middle) we add a 
selective defining event that does not change the label that becomes fixed
(i.e. the ancestor remains the same),
and that targets the newly chosen position $U^{(k)}_{}$. In contrast, in case 
(b) (cf. Figs. \ref{step1} and \ref{stepk}, right) the additional selective defining
event emanates from position $U^{(k)}_{}$ and hits the 
ancestor of the previous step (which ceases to be ancestor). The result is a
shifting of the ancestor `to the left', i.e. to position  $U^{(k)}_{}$. That is, each selective
defining event that goes back to case (b) gives rise to a new dotted line in Figs. 
\ref{step1} and \ref{stepk}.
To avoid misunderstandings, we would like to emphasise that the details of the
genealogies in the pictures are for the purpose of illustration only; the only thing we
really construct is the sequence of tuples $(\mathcal{I}^{(n)}_{}, \mathcal{J}_1^{(n)}, \dots , \mathcal{J}_{n}^{(n)})$.
\begin{figure}[ht]
\begin{center}
\begin{minipage}{0.315 \textwidth}
\psfrag{P}{\Large$\mathcal{I}^{(0)}_{}$}
\psfrag{1}{\Large$1$}
\begin{center}
\includegraphics[angle=0, width=3.5cm, height=4.2cm]{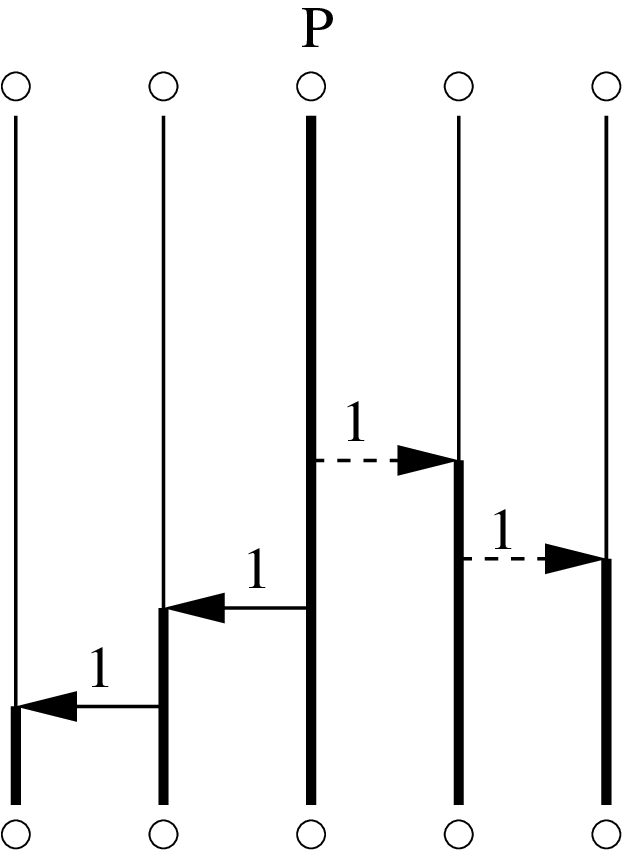}
\end{center}
\end{minipage}
\begin{minipage}{0.315 \textwidth}
\psfrag{P}{\Large$\mathcal{I}^{(1)}_{}$}
\psfrag{R}{\Large$\mathcal{J}^{(1)}_{1}$}
\psfrag{1}{\Large$1$}
\psfrag{s}{\Large$s_N^{}$}
\begin{center}
\includegraphics[angle=0, width=3.5cm, height=4.2cm]{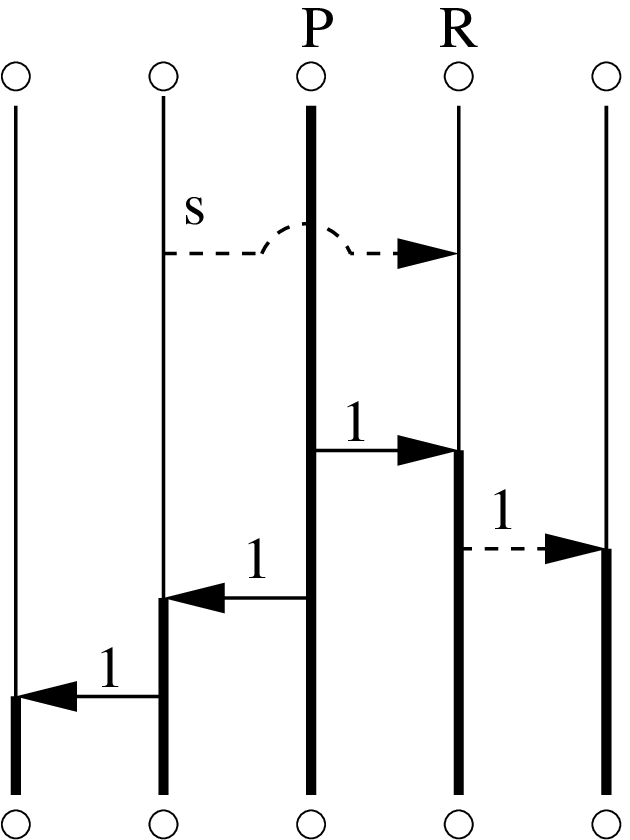}
\end{center}
\end{minipage}
\begin{minipage}{0.315 \textwidth}
\psfrag{P}{\Large$\mathcal{I}^{(1)}_{}$}
\psfrag{R}{\Large$\mathcal{J}^{(1)}_{1}$}
\psfrag{1}{\Large$1$}
\psfrag{s}{\Large$s_N^{}$}
\begin{center}
\includegraphics[angle=0, width=3.5cm, height=4.2cm]{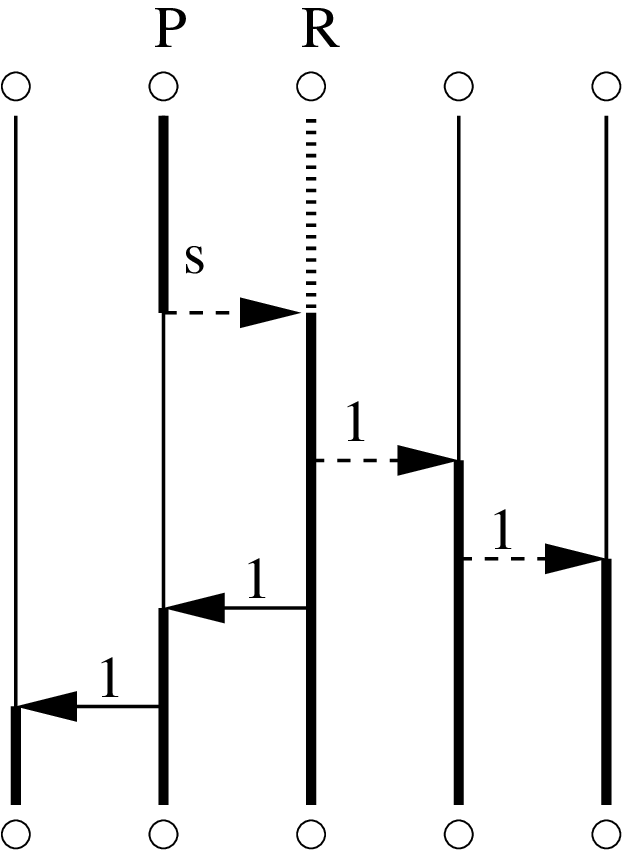}
\end{center}
\end{minipage}
\caption{Steps $0$ and $1$ of Algorithm \ref{algorithm} and corresponding genealogical interpretations.
Left: step $0$ (no selective defining events), middle: step $1$
(a), right: step $1$ (b) (each with one selective defining
event).} 
\label{step1}
\end{center}
\end{figure} 

\begin{figure}[ht]
\begin{center}
\begin{minipage}{0.315 \textwidth}
\psfrag{In}{\Large$\mathcal{I}^{(k-1)}_{}$}
\psfrag{Jn}{\Large$\mathcal{J}^{(k-1)}_{1}$}
\psfrag{Jp}{\Large$\mathcal{J}^{(k-1)}_{k-1}$}
\psfrag{1}{\Large$1$}
\psfrag{s}{\Large$s_N^{}$}
\begin{center}
\includegraphics[angle=0, width=4.3cm, height=4.5cm]{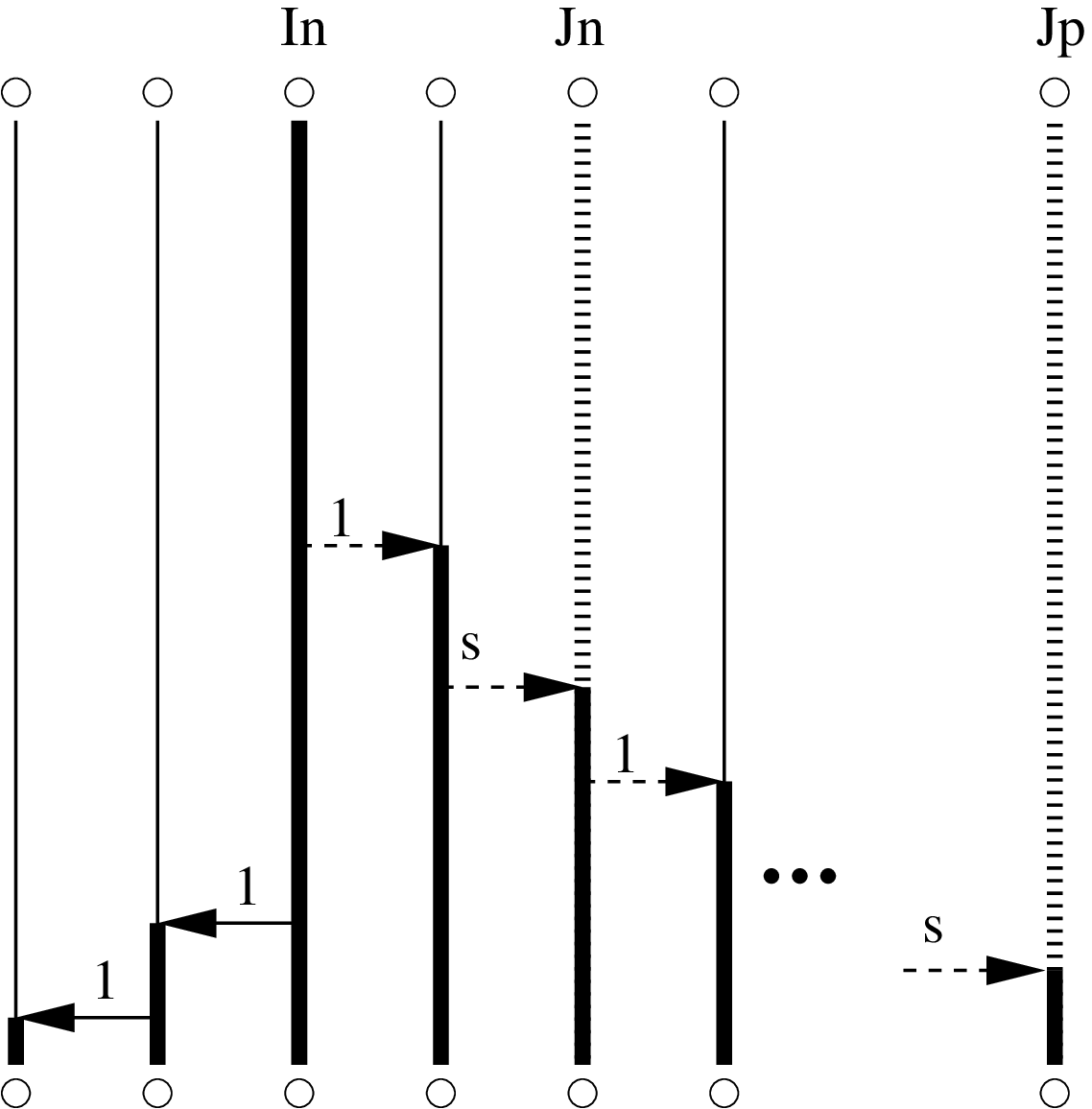}
\end{center}
\end{minipage}
\begin{minipage}{0.315 \textwidth}
\psfrag{In}{\Large$\mathcal{I}^{(k)}_{}$}
\psfrag{Jn}{\Large$\mathcal{J}^{(k)}_{1}$}
\psfrag{Jq}{\Large$\mathcal{J}^{(k)}_{2}$}
\psfrag{Jp}{\Large$\mathcal{J}^{(k)}_{k}$}
\psfrag{1}{\Large$1$}
\psfrag{s}{\Large$s_N^{}$}
\begin{center}
\includegraphics[angle=0, width=4.3cm, height=4.5cm]{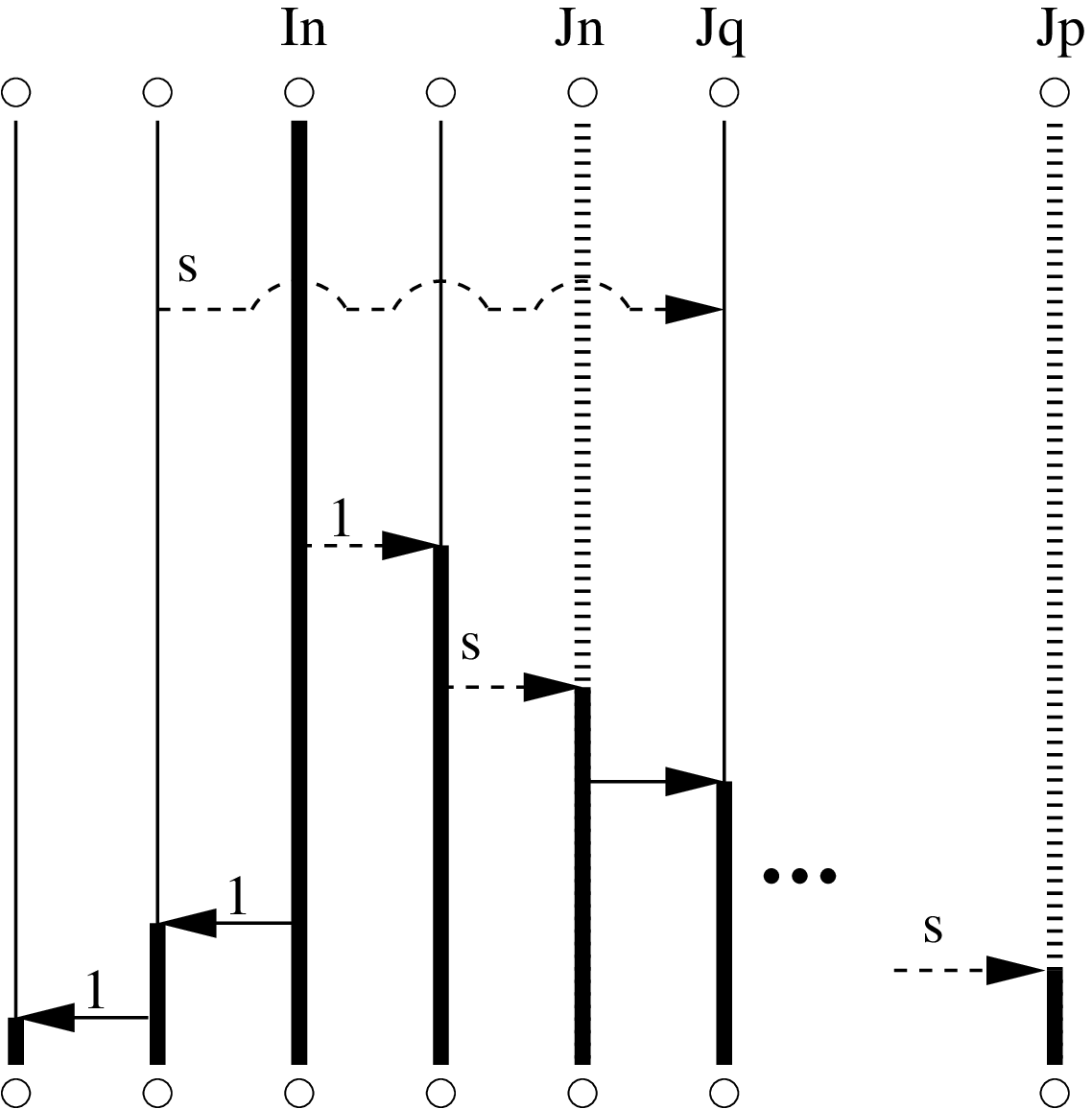}
\end{center}
\end{minipage}
\begin{minipage}{0.315 \textwidth}
\psfrag{Iv}{\Large$\mathcal{I}^{(k)}_{}$}
\psfrag{In}{\Large$\mathcal{J}^{(k)}_{1}$}
\psfrag{Jn}{\Large$\mathcal{J}^{(k)}_{2}$}
\psfrag{Jp}{\Large$\mathcal{J}^{(k)}_{k}$}
\psfrag{1}{\Large$1$}
\psfrag{s}{\Large$s_N^{}$}
\begin{center}
\includegraphics[angle=0, width=4.3cm, height=4.5cm]{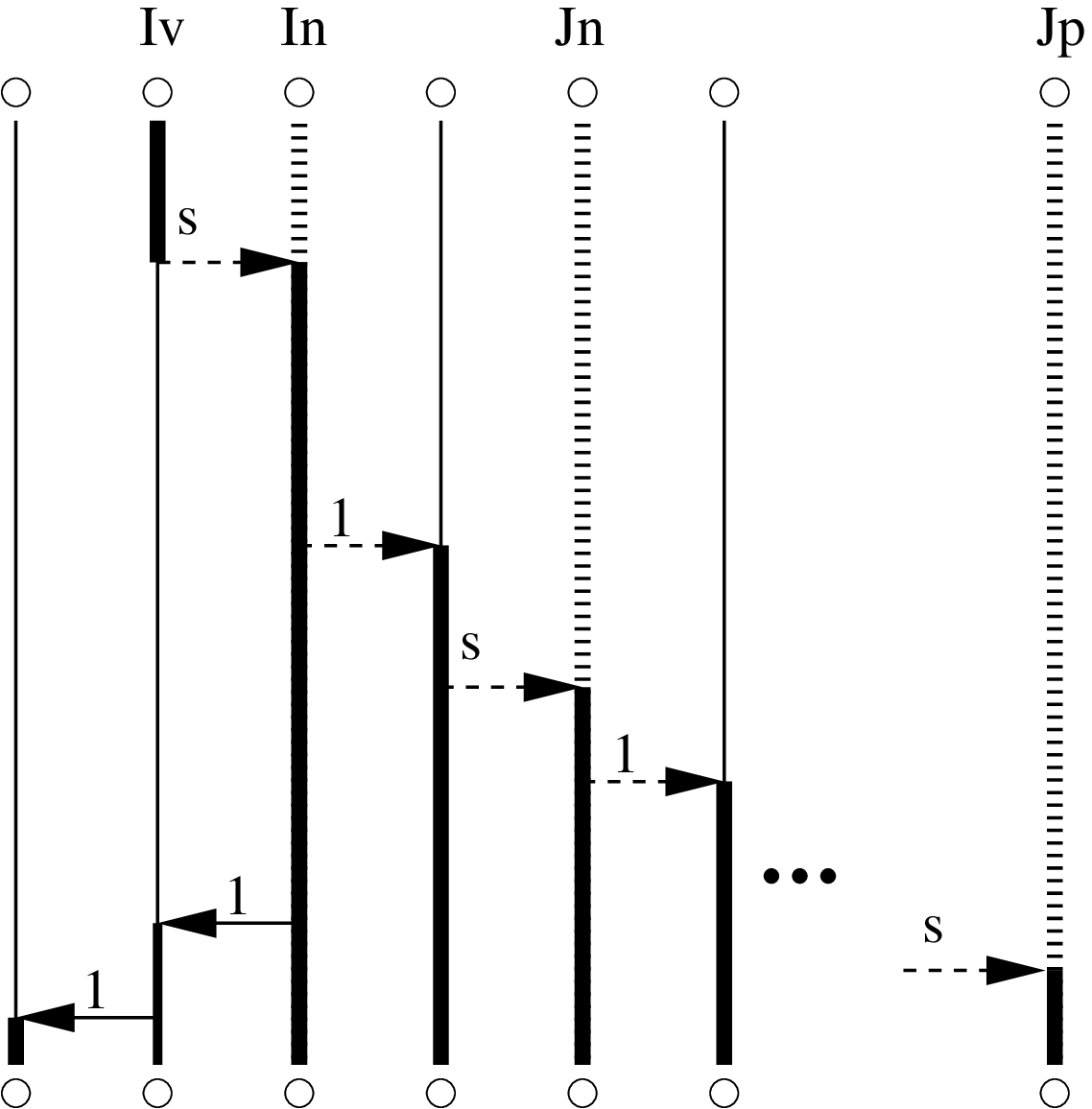}
\end{center}
\end{minipage}
\caption{Step $k$ of Algorithm \ref{algorithm} and corresponding genealogical interpretations. Situation
after step $k-1$ (left) and its modification according to step $k$ (a) (middle)
and step $k$ (b) (right).} 
\label{stepk}
\end{center}
\end{figure} 

We now have everything at hand to provide a genealogical interpretation for
the fixation probabilities
$h_i^N$ respectively $h(x)$ (cf. \eqref{h_i_alt} respectively \eqref{h(x)}). 
We have seen that the tuples $(\mathcal{I}^{(n)}_{},\mathcal{J}^{(n)}_{1}, \dots , \mathcal{J}^{(n)}_{n})$
are constructed such that
\begin{equation*}
 \mathbb{P}(\mathcal{I}^{(n)}_{} = i,\mathcal{J}^{(n)}_{1}= j_1^{}, \dots , \mathcal{J}^{(n)}_{n} = j_n^{})
 = \mathbb{P} (I^N_{}=i, J_1^N = j_1^{}, \dots , J_n^N = j_n^{} \mid Y^N_{} = n)
\end{equation*}
for all $n \geq 0$ and $j_1^{}, \dots , j_n^{} \in \{ i+1, \dots , N \}$, and, via marginalisation,
\begin{equation*}
 \mathbb{P}(\mathcal{I}^{(n)}_{} =i) = \mathbb{P} (I^N_{} = i \mid Y^N_{} =n).
\end{equation*}
We reformulate the decomposition of \eqref{h_i alternativ} to obtain
\begin{align}
h_i^{N} 
&=\mathbb{P} ( I^N_{} \leq i \mid Y^N_{}=0 ) \mathbb{P} ( Y^N_{} \geq 0 ) \nonumber  \\
& \hphantom{=} 
+ \sum_{n = 1}^{N-1} \left[ \mathbb{P} ( I^N_{} \leq i \mid Y^N_{}=n )-\mathbb{P} ( I^N_{} \leq i \mid Y^N_{}=n-1 ) \right]  \mathbb{P} ( Y^N_{} \geq n ) \nonumber \displaybreak[0] \\ 
&= \mathbb{P} (\mathcal{I}^{(0)}_{} \leq i) \mathbb{P} ( Y^N_{} \geq 0 ) 
+ \sum_{n = 1}^{N-1} \left[ \mathbb{P} ( \mathcal{I}^{(n)}_{} \leq i )-\mathbb{P} (\mathcal{I}^{(n-1)}_{} \leq i ) \right]  \mathbb{P} ( Y^N_{} \geq n ) \nonumber \displaybreak[0] \\
&=\mathbb{P} (\mathcal{I}^{(0)}_{} \leq i)\mathbb{P} ( Y^N_{} \geq 0 )
+ \sum_{n = 1}^{N-1} \mathbb{P}(\mathcal{I}^{(n)}_{} \leq i, \mathcal{I}^{(n-1)}_{} > i) \mathbb{P} ( Y^N_{} \geq n ),
\label{h_i simulation}
\end{align}
where the last equality is due to the fact that the ancestor's label is non-increasing in $n$.
In \eqref{h_i simulation}, the fixation probability $h_i^N$ is thus decomposed according to the first step in the 
algorithm in which the ancestor has a label in $\{ 1, \dots, i \}$. The
probability that this event takes place in step $n$ may, in view of the simulation
algorithm, be expressed explicitly as
\begin{align}
 \mathbb{P}(\mathcal{I}^{(n)}_{} \leq i, \mathcal{I}^{(n-1)}_{} > i)
 &=\mathbb{P}(\mathcal{I}^{(n)}_{} \leq i, \mathcal{I}^{(0)}_{} , \mathcal{I}^{(1)}_{} , \dots ,   \mathcal{I}^{(n-1)}_{} > i) \displaybreak[0] \nonumber \\
 &=\mathbb{P}(U^{(n)}_{} \leq i, U^{(0)}_{} , U^{(1)}_{} , \dots ,   U^{(n-1)}_{} > i) \displaybreak[0] \nonumber \\
 &=\frac{i}{N} \prod_{j=0}^{n-1} \frac{N-i-j}{N-1-j}.
 \label{diskrete sampling distribution}
\end{align}
Together with \eqref{h_i simulation} and \eqref{a_n^N Interpretation} 
this provides us with a term-by-term interpretation of \eqref{h_i_alt}.

In particular, \eqref{diskrete sampling distribution} implies that
\begin{equation*}
 \lim_{N \to \infty} \mathbb{P}(\mathcal{I}^{(n)}_{} \leq i_N^{}, \mathcal{I}^{(n-1)}_{} > i_N^{})
 = x (1-x)^n
\end{equation*}
for a sequence  $(i_N^{})_{N \in \mathbb{N}}$ with $i_N \in \{ 1, \dots , N \}$ and 
$\lim_{N \to \infty} i_N^{} / N=x$, $x \in [0,1]$. Thus,
in the diffusion limit, and with $a_n^{}= \lim_{N \to \infty} \mathbb{P}(Y^N_{} \geq n)$ of  
\eqref{a_n Interpretation}, \eqref{h_i simulation} turns into 
\begin{equation*}
 h(x) = \lim_{N \to \infty} h_{i_N^{}}^N = x + \sum_{n \geq 1} x (1-x)^n a_n^{},
\end{equation*}
which is the representation \eqref{h(x)}, and
which is easily checked to coincide with \eqref{h diffusion} (as obtained by
the direct approach of Sec. \ref{Diffusion limit of the labelled Moran model}).

\section{Discussion}
\label{Discussion}
We have reanalysed the process of fixation in a Moran model with two types
and (fertility) selection by means of 
the labelled Moran model. It is interesting to compare the labelled Moran model with the corresponding lookdown construction: In the lookdown
with fertility selection, \emph{neutral} arrows only point in one
direction (from lower to higher levels), whereas \emph{selective} arrows may 
appear between arbitrary levels. In contrast, the labelled Moran model contains 
\emph{neutral arrows in all directions}, but \emph{selective arrows} only occur from
lower to higher labels. Also, the labelled Moran model deliberately dispenses with
exchangeability, which is an essential ingredient of the lookdown. It may be 
possible to transform the labelled Moran model into a lookdown (by way of random permutations),
but this remains to be verified. 

We certainly do not advertise the labelled Moran model as a general-purpose tool;
in particular, due to its arbitrary neutral arrows, it does not allow the construction of
a sequence of models with increasing $N$ on the same probability space.
However,  it turned out to be particularly useful for the purpose considered in this paper;
the reason seems to be that the fixation probabilities of its individuals, $\eta_i^N$,
coincide with the individual terms in \eqref{h_i}. Likewise, we obtained a
term-by-term interpretation of the fixation
probability \eqref{h_i_alt}  as a decomposition according to
the number of selective defining events, which may successively shift the
ancestral line to the left, thus placing more weight on the `more
fit' individuals. 

Indeed, the incentive for this paper was the intriguing representation \eqref{h(x)}
of the fixation probabilities, which was first observed by Fearnhead (2002)
in the context of (a pruned version of) the ASG, for the diffusion limit of the Moran
model with mutation and selection at stationarity, and further investigated by Taylor (2007).
Quite remarkably, \eqref{h(x)} carries over to the case without mutation (which turns
the stationary Markov chain into an absorbing one), with Fearnhead's $a_n^{}$ reducing to ours
(see Kluth et al. (2013) and Sec. \ref{The Moran model with selection}). The observation that
our $W^{\infty}_{}$, that is, $1$ plus the number of selective defining events in the diffusion
limit, has the same distribution as the number of branches in the ASG (see Sec. \ref{Diffusion limit of the labelled Moran model})
fits nicely into this context, but still requires some further investigation.
Needless to say, the next challenge will be to extend the results to the case with mutation -- this
does not seem to be straightforward but is quite possible, given the tools and insights that
have become available lately.

\section*{Acknowldegments}
It is our pleasure to thank Anton Wakolbinger, Ute Lenz, and Alison Etheridge for
enlightening discussions. This project received financial support from
Deutsche Forschungsgemeinschaft (Priority Programme SPP 1590 `Probabilistic
Structures in Evolution', grant no. BA 2469/5-1).

\end{document}